\documentclass[accepted]{uai2023} 

\usepackage[american]{babel}

\usepackage{natbib} 
\usepackage{mathtools} 
\usepackage{booktabs} 
\usepackage{tikz} 

\usepackage{subfig}
\usepackage{comment}
\usepackage{amsmath,amssymb,amsfonts,amsthm}
\usepackage{algorithmic}
\usepackage{algorithm}
\usepackage{wrapfig}
\usepackage{dsfont}
\usepackage{multirow}
\usepackage{tabularx}

\newtheorem{proposition}{{\bf Proposition}}
\newtheorem{remark}{{\bf Remark}}

\title{Inference and Sampling of Point Processes from Diffusion Excursions}

\author[1,2]{\href{mailto:<ali.hasan@duke.edu>?Subject=Your UAI 2022 paper}{Ali~Hasan}{}}
\author[2]{Yu~Chen}
\author[1]{Yuting~Ng}
\author[3]{Mohamed Abdelghani}
\author[2]{Anderson~Schneider}
\author[1]{Vahid~Tarokh}
\affil[1]{%
    Department of Electrical and Computer Engineering\\
    Duke University\\
    Durham, North Carolina, USA
}
\affil[2]{%
    Machine Learning Research\\
    Morgan Stanley
}

\affil[3]{%
Department of Mathematics\\
    University of Alberta\\
    Edmonton, Alberta, Canada
}

\usepackage{xr}
\makeatletter

\newcommand*{\addFileDependency}[1]{
\typeout{(#1)}
\@addtofilelist{#1}
\IfFileExists{#1}{}{\typeout{No file #1.}}
}\makeatother

\newcommand*{\myexternaldocument}[1]{%
\externaldocument{#1}%
\addFileDependency{#1.tex}%
\addFileDependency{#1.aux}%
}

\myexternaldocument{uai2023-supplement}

\begin{document}
\maketitle

\begin{abstract}
Point processes often have a natural interpretation with respect to a continuous process.
We propose a point process construction that describes arrival time observations in terms of the state of a latent diffusion process.
In this framework, we relate the return times of a diffusion in a continuous path space to new arrivals of the point process. 
This leads to a continuous sample path that is used to describe the underlying mechanism generating the arrival distribution.
These models arise in many disciplines, such as financial settings where actions in a market are determined by a hidden continuous price or in neuroscience where a latent stimulus generates spike trains.
Based on the developments in It\^o's excursion theory, we propose methods for inferring and sampling from the point process derived from the latent diffusion process. 
We illustrate the approach with numerical examples using both simulated and real data.
The proposed methods and framework provide a basis for interpreting point processes through the lens of diffusions. 

\end{abstract}

\section{Introduction}
Point processes are a powerful modeling tool for describing patterns of arrivals, with applications ranging from environmental and biological sciences to financial markets and social behavior~\citep{bjork1997bond, rizoiu2017hawkes, subramanian2020point,stoyan2000recent}. 
Often, a point process is represented through an \emph{intensity function}, which is a function that describes the expected number of arrivals.
This function is the primary mechanism for interpreting the properties of the process, with standard models such as the Poisson process and the Hawkes process as primary examples. 
However, considering only the intensity function may not provide a complete understanding of the underlying cause of arrivals of points

\begin{figure}
\centering
\includegraphics[width=0.48\textwidth]{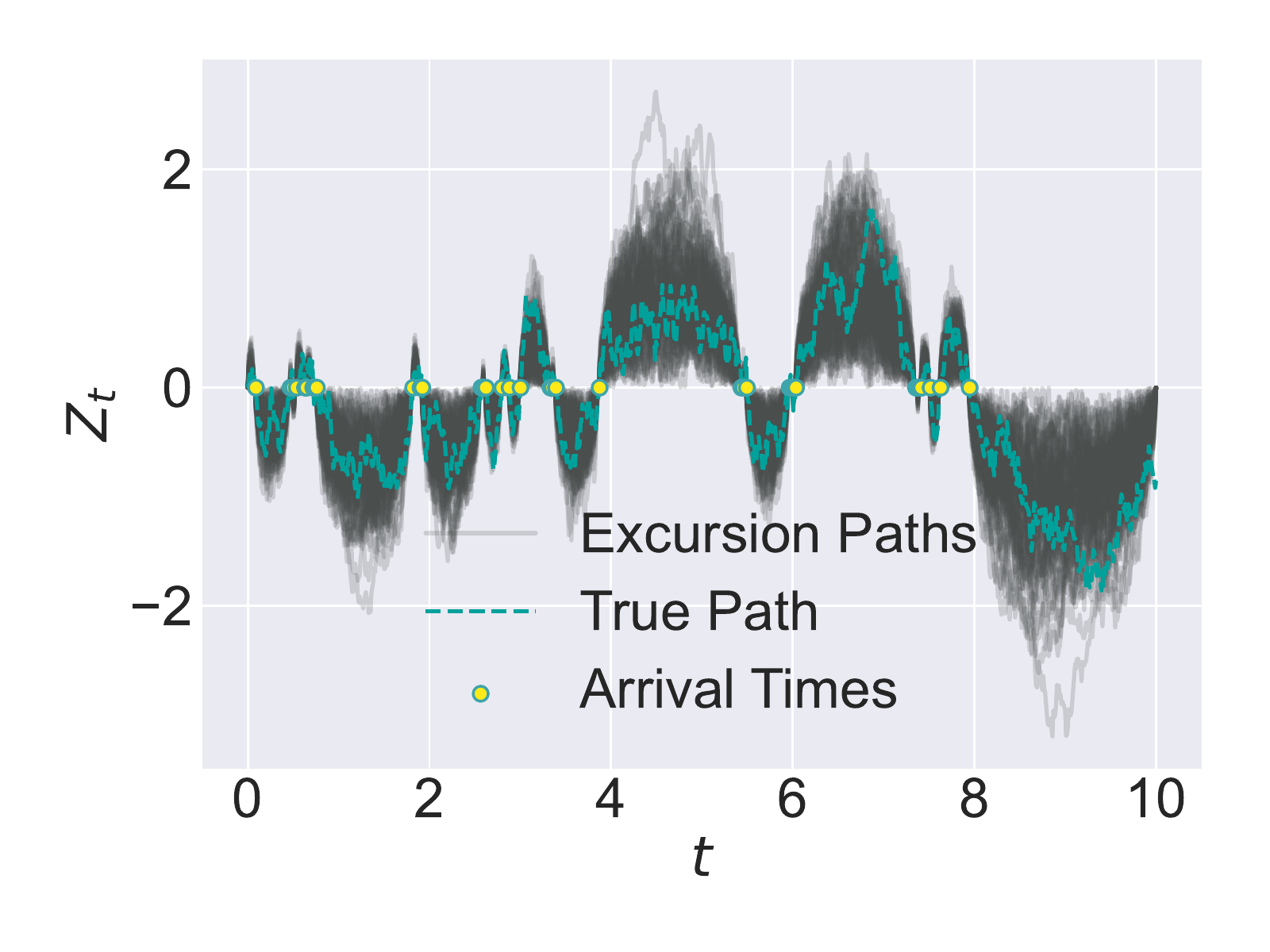}
\caption{Decomposition of a sample path into distributions of Brownian excursions. Dashed line represents the true signal and solid lines represent possible excursions between observations of arrival times indicated by circle markers.}
\label{fig:decomp}
\vspace{-10pt}
\end{figure}

In many cases, a point process may be related to a \emph{continuous process}. 
The choice of the continuous process is also motivated by applications. 
To list a few examples, neuron spike trains may be related to the first passage time of an underlying chemical concentration surpassing a threshold~\citep{sacerdote2003threshold}.
Similarly, intracellular events are considered to be a function of a protein concentration exceeding a threshold and bursty transcription relates to the continuous movement of underlying molecules~\citep{ghusinga2017first, lammers2020matter}.
In an economic setting, one can think of information flow in a market, and the ensuing point process generated by orders of agents on an exchange as a function of the information flow~\citep{babus2018trading}.
All of these models consider a multi-scale approach such that the point process is generated as a function of the unobservable continuous process.
Developing inference methods for recovering a possible continuous process could provide additional insights into the point process being studied. 
This leads us to the motivation of this work, where we focus on continuous stochastic processes defined by It\^o diffusions and relate these paths to the arrival times of the point process. 
Specifically, we consider a decomposition into paths known as \emph{excursions} -- paths that begin and end at a reference state and are constrained to stay above or below the reference state for their entirety. The length of an excursion correspond to an interarrival time of the point process.
This decomposition is illustrated in Figure~\ref{fig:decomp} where the sample path is decomposed into excursions from the reference state of 0, marked by arrivals of the point process.
The original idea for constructing such a process was introduced by~\citet{ito1972poisson}, where the decomposition of sample paths into excursions was used as an alternative tool to stochastic calculus for studying diffusions.
It\^o additionally raised the question of what continuous process could represent observations of arrival times, which is what we contribute towards in this work~\citep{watanabe2010ito}. 

\paragraph{Motivating Example}
Consider events in a market where individuals are buying and selling a group of assets by placing bids and asks.
We can model this as a marked point process where the mark is given by the type of action (bids or asks) and the price at which the asset is requested.
Then, we can assume that the different arrivals are associated with an unobserved, fair price that governs the asset --- bids are placed when the asset is below its fair price and asks are placed when the asset is above the fair price. 
These \emph{excursions} above and below the fair price give rise to the point process structure, and that is what we aim to model through this framework. 
The multi-dimensional case then corresponds to the point process of multiple, correlated assets in a market.
The parameters recovered then correspond to the diffusion that models the unobserved fair price.

\subsection{Related Work}

A number of research directions are related to understanding point processes through continuous processes, and, in the particular case of Brownian motion, the most relevant is the study of first hitting times (FHT).
The FHT problem was originally posed by Albert Shiryaev regarding whether there exists a boundary such that the stopping time of a Brownian motion at such a boundary is distributed according to an exponential distribution~\citep{potiron2021existence}.
\citet{anulova1981markov} answered the question affirmatively for a series of barriers but did not make any consideration of the regularity of the barriers. 
Theoretical investigation of this problem has led to numerous computational approaches for computing the FHT density for arbitrary boundaries (e.g.~\citep{jaimungal2014generalized,zucca2009inverse}). 
A feasible computational solution for estimating the drift of a process with a desired FHT density was provided by~\citet{ichiba2011efficient} who described a representation for the distribution in terms of an expectation but did not consider the problem of estimation. 
While these works consider methods of finding the appropriate boundary, they deviate from our goal of recovering a \emph{continuous} latent process due to the fact that the state of the process must be reset after the arrival of each point.
Since excursions begin and end at the same location, studying the excursion distribution allows one to reconstruct the full continuous path. 
We discuss the similarities between the proposed method and the FHT problem in greater detail in Section~\ref{sec:fht}.

The second relevant line of work is based on It\^o's description of Brownian paths through excursions~\citep{ito2020poisson}.
Using this framework, It\^o described point processes over the space of excursions as a technique for analyzing properties of diffusions.
\citet{watanabe1987construction} considered a mathematical construction of semimartingales through their excursions. 
Recent approaches considered further analyses and applications in finance~\citep{ananova2020excursion}.
Closely related is the Az\'ema martingale which provides an estimate of the value of a Brownian motion when only observing the sign of its excursion~\citep{ccetin2012filtered}.
This has found applications in pricing Parisian options and estimating firm default risk~\citep{ccetin2012filtered}.
See~\citet{watanabe2010ito} for a more comprehensive history on excursions and the development of It\^o's excursion theory. 
However, the related question regarding estimating a diffusion from its excursion lengths has not yet been answered through a computationally feasible framework, which is the main purpose of this work.
A closely related line of work concerns filtering problems where the observation is a point process with an intensity function given by a diffusion. 
A Cox process can be understood as a Poisson process with a stochastic intensity function. 
\citet{jaiswal2020variational} described a method for computing the posterior distribution of a Cox process with intensity given by diffusion through solving a stochastic PDE.
These methods can also be seen as filtering problems where the observation model is a point process with a latent continuous process.
Unfortunately, the approach generally requires solving computationally intractable equations for computing the posterior. 
To mitigate this issue~\citet{lloyd2015variational} describes modulating the intensity function with a Gaussian process and describes a variational inference approach for optimizing the parameters.

Other applications of these ideas have also been considered, particularly in the case of the FHT.
In survival modeling, the first hitting time of a diffusion at some region of the domain is used to determine the end of the life of a particular process.
\citet{roberts2010latent} proposes a method for recovering diffusions based on survival data with the assumption of underlying diffusion. 
The approach is based on a Markov-Chain Monte Carlo method for estimating the posterior density given the survival times and the hazard function is estimated with respect to the diffusion.
The continuous process can represent the state of, for example, an engine throughout its lifetime, and can be useful for gaining an interpretation of the stresses faced by the engine.
\citet{maystretemporally} considered modeling the survival distribution in terms of the first hitting time of a discrete time Markov process and related this procedure to a temporal difference learning problem in reinforcement learning. 

A further distinction between intensity-based models with stochastic intensity and excursion-based models is that given the underlying continuous process, the time stamps are sampled randomly in the former, while the time stamps are deterministic in the latter.
In other words, the intensity-based model usually involves \textit{doubly stochastic process}, while the excursion-based model does not.

\subsection{Contributions}
We propose a modeling framework based on It\^o's excursion theory that represents a point process over the line as a decomposition of a diffusion in terms of excursions where the excursion length corresponds to the interarrival time.
Our contributions are then the following:
\begin{enumerate}
    \item 
    We extend the point process modeling framework based on the diffusion process, where the time stamps are determined by excursions; 
    
    \item
    We provide an inference algorithm for the model; 
    
    \item 
    We demonstrate the versatility of the framework by presenting applications to many classes of distributions. 

    \item We illustrate the framework's utility and interpretability on a variety of synthetic and real data experiments.

\end{enumerate}

\section{Background}
To provide the initial exposition, we will assume our class of continuous processes are solutions of the one-dimensional stochastic differential equations (SDEs) driven by Wiener processes. 
We suppose that the latent process $Z_t$ is the solution to the SDE given by
\begin{equation}
    \label{eq:sde}
\mathrm{d} Z_t = \mu(Z_t, t) \mathrm{d} t + \sigma(t) \mathrm{d} W_t
\end{equation}
where $W_t$ is a standard Brownian motion.
The object of interest in this work is to model the drift function $\mu$.

\subsection{Definitions}
Here we provide the definitions that link the point process to the diffusion in~\eqref{eq:sde}.
Our overall goal in estimation is to find a $\mu$ such that the corresponding excursion length distribution is the same as the interarrival time distribution of the point process. 
Therefore, we will study the properties of excursions to describe the method.
We follow the terminology of~\citet{pitman2007ito} to introduce the definitions and defer to that manuscript for a more comprehensive study on the implications of Brownian excursions.

Consider a sample path $Z_t$ satisfying~\eqref{eq:sde}.
An excursion set can be thought of as the subsets of $Z_t$ that exceed a particular function $f(t)$.
The length of the excursion is then related to the times that $Z_t$ first hits and surpasses $f(t)$ and the time that $Z_t$ returns to $f(t)$.
Define the set of hitting times by
\begin{equation}
\mathbb{H}_t := \left\{\sup_{r \in [0, s] } \{ r \mid Z_r = f(r) \} \: \bigg |  \: s < t \right \}
\label{eq:zeros}
\end{equation}
and then consider the \emph{local time} at $f(t)$ as
$$
L_t = \lim_{\epsilon \to 0 }\int_0^t\frac1{2\epsilon} \mathds{1}_{|Z_s - f(s)| < \epsilon}\mathrm{d}s.
$$
The local time is an increasing function that, heuristically, describes the amount of time the process $Z_s, s < t$ has spent at $f(s), s < t$ up to time $t$.
We next define the \emph{inverse local time} as
$\tau_\ell = \inf\{t >0 : L_t > \ell \}, \ell \geq 0$ which describes the time at which $Z_t$ has spent $\ell$ time at $f(t)$.
An excursion straddling $(\tau_{\ell^-}, \tau_{\ell})$ is then defined as 
\begin{equation}
    e^\ell := \{Z_s : s \in ( \tau_{\ell^-}, \tau_{\ell}) \}.
\end{equation}
where $\tau_{\ell^-}$ is the left-sided limit of the inverse local time.
We note that the space of all excursions is not relevant for our purposes of modeling due to the topological properties of~\eqref{eq:zeros}.
To give an example, taking $Z_t$ as standard Brownian motion starting at zero and $f(t)=0$ results in $\mathbb{H}_1$ being a perfect set with properties that are not practical for a modeling task.
From an applied perspective, very small excursions would not be observed due to limitations on the resolution of measuring devices used to collect data.
Instead, one usually considers a subset of excursion paths that have some relevance, such as excursions of minimum length or minimum height\footnote{This is the interpretation of the excursion measure given by D. Williams, see~\citet[Chapter 6]{yen2013local} for a detailed description.}.
In~\citet{ananova2020excursion}, excursions reaching a minimum height of $\delta$, described as $\delta$-excursions, were considered. 
Under this construction, $e^\ell$ allows us to decompose continuous sample paths given by $Z_t$ into different excursions with excursion lengths indexed by $\ell$.
This generates a Poisson process where excursion lengths define the interarrival times of the point process.

To illustrate these concepts, we again refer to  Figure~\ref{fig:decomp} where the original sample path representing $Z_t$ (blue) is decomposed into excursions above and below the line $f(t) = 0$.
The arrival times (yellow circles) describe the end of an excursion, the last point in $\mathbb{H}_t$. 
Finally, multiple samples of excursions (grey) with length $\tau_{\ell} -\tau_{\ell^-}$ are illustrated to describe the relationship to the true excursion of $Z_t$.
For the remainder of the text, we will suppose that $f(t) = 0 $ for all $t$ and consider the set of times $\mathbb{H}_t$ when $Z_t$ returns to 0.

\subsection{Assumptions }
We state a few more properties to ensure we can compute valid excursion densities.
These are conditions on the drift $\mu$ so that the interarrival time of the excursion is finite, which in turn guarantees that the measure is a valid density.
\begin{enumerate}
    \item The diffusion must be recurrent; i.e. $\mathbb{P}(\tau = \infty) = 0$. This is guaranteed if $\lim_{a\to \infty} S(a,t) = \infty$ and  $\lim_{a\to -\infty} S(a,t) = -\infty$ where 
    $$
    S(a,t) = \int_0^a \exp \left( \int_0^b \frac{-2 \mu(x,t)}{\sigma(x,t)} \mathrm{d}x \right) \mathrm{d}b
    $$
    for all $t$, 
    
    \item The measure counting the number of excursions must be finite; that is, we do not consider excursions that have negligible length. 
    
    \item Novikov's condition
    $\mathbb{E}[e^{\frac12 \int_0^\tau |Z_t|^2 \mathrm{d} t} ] < \infty$ for Girsanov's theorem to hold.
    
    \item Existence of a $t$-continuous strong solution to \eqref{eq:sde}, that is $\mu, \sigma$ are Lipshitz \citep[Theorem 5.2.1]{oksendal2003stochastic}.
    
\end{enumerate}

A final assumption that we consider is that $\sigma(t) = 1$.
This is not necessary, but as noted later in the text, by introducing an additional parameter $\delta$ regarding the minimum height of the excursion, there exists an estimation ambiguity between the $\sigma$ and $\delta$ parameters.
To circumvent the estimation ambiguity, we consider recovering the transformed process given by the Lamperti transform which results in the diffusion with unit volatility. 
Additional details are presented in Appendix~\ref{sec:lamperti}.
We also note that the drift $\mu$ can depend on history or on an additional process, but we leave the drift in its standard form for ease of exposition.

\section{Method}
We now describe the inference method for finding $\mu$ given a set of interarrival times.
We define the interarrival times as the set $\mathbb{T}_t = \{\tau_1, \ldots, \tau_N\}$ and relate them to the set $\mathbb{H}$ by $\tau_i = \mathbb{H}^{(i+1)}_t - \mathbb{H}^{(i)}_t$, $\mathbb{H}^{(i)}_t$ being the $i^\text{th}$ arrival time in ascending order up to time $t$. 
For example, $\mathbb{T}_t$ would contain elements that are exponentially distributed in the case of a Poisson process.
As mentioned in the previous section, we remove small excursions by only considering excursions with a minimum height by redefining $\mathbb{H}_{t, \delta} = \{ \tau_i \in \mathbb{H}_t \; | \; \max_{s\in (\tau_{i-1} , \tau_i)} Z_s - f(s) \geq \delta \}$ in~ \eqref{eq:zeros} for some $\delta > 0$.
We consider a minimum height so that the density remains absolutely continuous with respect to the Lebesgue measure on the positive real line.
To outline the method, we first state the excursion length distribution of standard Brownian motion with minimum height $\delta$.
We then perform a change of measure to find the excursion length distribution of a diffusion with drift given by $\mu$.
We represent the drift $\mu$ by a neural network and optimize for its parameters via stochastic gradient descent and maximum likelihood estimation on observations of excursion lengths given by interarrival times. 
In the remaining of the text, we will denote a general excursion as $e$ and the excursion at time $t$ as $e_t$.

\subsection{Excursion Length Density of Brownian Motion}
Excursion times from 0 to $\delta$ and back to 0 have the distribution given by the inverse Laplace transform of
$
\mathbb{E}[e^{-\lambda \tau}] = e^{-2\sqrt{2 \lambda} \delta}.
$
Taking the inverse Laplace transform, we obtain a zero shifted L\'evy distribution with scale parameter as $4 \delta^2$ and PDF of
\begin{equation}p_e(\tau; \delta) = \delta \sqrt{\frac2{\pi \tau^3}}\exp{\left(-\frac{2\delta^2}{\tau}\right)} .
\label{eq:hitting_bm}
\end{equation}
Additional details regarding this derivation are provided in Appendix~\ref{sec:laplace}.
Note that if $\sigma$ is included in the density of~\eqref{eq:hitting_bm} then $\delta$ and $\sigma$ are unidentifiable, motivating the previously stated assumption on $\sigma=1$. 
In some cases, it may be easier to optimize one than the other, e.g. when simulating excursions with variance $\sigma$ is easier than excursions of a minimum height $\delta$, but we will focus on a minimum height $\delta$.
With the Brownian excursion length density in mind, we now consider a change of measure for the drifted case.

\subsection{Change of Measure for Excursion Length Density of Diffusion}
We consider an approach inspired by~\citet[Section A.1]{ichiba2011efficient} where the authors use a change of measure technique to compute the density of the FHT of a diffusion. 
Let $e_t$ be the value of a Brownian excursion of length $\tau$ at time $t$.
The excursion length density of a diffusion follows an expectation of a Radon-Nikodym derivative between the base measure on the space of $\delta$-excursions $\mathbb{Q}_\updownarrow^\delta$ and the diffusion measure $\mathbb{P}^\mu$.

\begin{proposition}[Diffusion Excursion Density]
Let $Z_t$ satisfy an SDE with drift $\mu$ such that $Z_t$ is recurrent at zero.
Then the density of the excursion lengths $\tau$ of $Z_t$ is given by:
\begin{align}
\nonumber &p_{Z}(\tau) = p_{e} (\tau; \delta) \times \\ &\mathbb{E}_{\mathbb{Q}_\updownarrow^\delta}\Bigg[\exp\bigg(\int_0^\tau \mu(e_t,t;\theta)\mathrm{d}e_t  - \frac12 \int_0^\tau \mu^2(e_t,t;\theta) \mathrm{d}t \bigg) \Bigg].
\label{eq:girsanov_one_dim}
\end{align}

\end{proposition}

\begin{proof}[Intuition of Proof]
    The proof follows a change of measure argument. 
    The full proof is in Appendix~\ref{sec:proof_com}.
\end{proof}

During optimization, $p_e$ does not need to be computed since it is a constant with respect to the input data $Z_t$.
However, if we consider $\delta$ as a parameter for the optimization, we may include it in the computation of the likelihood. 
As noted, we may also consider non-unit $\sigma$, but this comes at the expense of identifiability of $\delta$. 
This results in a modification of the expectation over sample paths which should have the corresponding $\sigma$.   
Finally, we require that the drift must be recurrent for the density to integrate to 1. 
Numerically, we enforce this condition by adding a regularizer that constrains the density to approximately integrate to 1.

\paragraph{Evidence Lower Bound}
Following Jensen's inequality, for maximizing~\eqref{eq:girsanov_one_dim} for given data, we can also optimize
\begin{align}
\nonumber \log p_z(\tau; \delta) 
& \geq \log p_e(\tau; \delta) + \\ & \mathbb{E}_{\mathbb{Q}_\updownarrow^\delta}\left[\int_0^\tau \mu(e_t,t;\theta)\mathrm{d}e_t - \frac12 \int_0^\tau \mu^2(e_t,t;\theta) \mathrm{d}t \right] \label{eq:elbo}
\end{align}

and therefore we need only to maximize~\eqref{eq:elbo} rather than~\eqref{eq:girsanov_one_dim}.
This may be beneficial in scenarios where numerical errors may make the calculation of the exponential unstable.
We detail an algorithm for estimating the drift $\mu$ from data in Algorithm~\ref{alg:mle}.

\begin{algorithm}[!ht]
	\caption{Inference for latent diffusion from arrival times}
	\label{alg:mle}
	\begin{algorithmic}[1]
	\STATE \textbf{Input:} Sequences of interarrival times: $\mathbb{T} = \left\{ \tau_1^{(j)}, \tau_2^{(j)}, \ldots, \tau_{n_j}^{(j)}\right\}_{j=1}^N$
	\STATE \textbf{Initialize Parameters:} Parameters of drift $\mu(x,t)$, step size $\Delta t$, total time $T$, initial state $X_0$, minimum height $\delta$ and variance $\sigma$, number of points in expectation $K$.
	\STATE Sample $K$ Brownian excursions using Vervaat transform~\citep{vervaat1979relation} of a Brownian bridge and the Euler-Maruyama method $\mathbb{E} = \{e_i\}_{i=1}^K$.
	\STATE Filter $\mathbb{E}$  by discarding $e_i$ where $\max e_i < \delta$.
	\STATE Numerically compute~\eqref{eq:elbo} for the data $\mathbb{T}$ with excursions $\mathbb{E}$ computed above. 
    \STATE Maximize~\eqref{eq:elbo} using gradient decent with respect to the parameters of $\mu$ and $\delta$.
    \STATE Repeat for $N$ iterations
	\end{algorithmic} 
\end{algorithm}

\subsection{Multidimensional Processes}
\label{sec:multi}
\begin{figure}
    \centering
    \includegraphics[trim=150pt 0pt 100pt 0pt, width=0.4\textwidth]{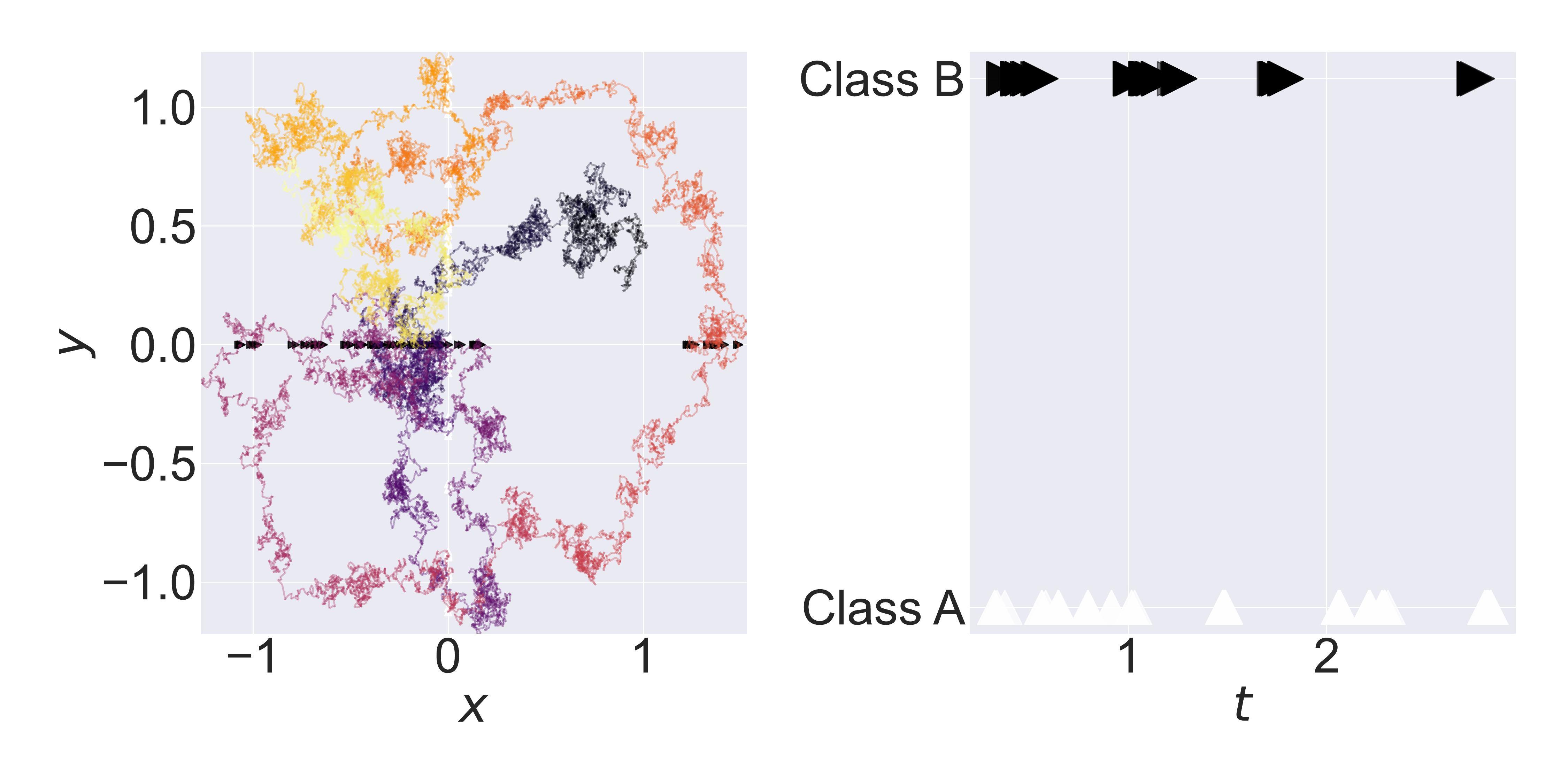}
    \caption{Schematic of the multidimensional framework. Left: the 2-dimensional latent diffusion crosses the axes producing points in the point process. Right: arrival times of the point process generated by the diffusion process on the left.
    Up triangles represent excursions from the $y$-axis and left triangles represent excursions from the $x$-axis.
    }
\label{fig:schematic}
\end{figure}
The approach extends to multidimensional diffusions that have interacting components, where we consider excursions away from each axis. 
The idea is illustrated in Figure~\ref{fig:schematic}, where a diffusion in the 2-D plane generates excursions from either axis, each axis corresponding to a different mark.
This leads to a dependence between the two classes through the drift function.
Our approach for calculating the corresponding $\mu$ from the data is analagous to the 1-dimensional case where we compute expectations over Brownian excursions with the same lengths as the interarrival time observations and repeat over each dimension.
In this case, $Z_t$ is a $d$-dimensional diffusion.
From there, the likelihood of the data is maximized using a single multi-dimensional drift function that governs the relationships between the different marked processes.
To estimate the drift from observations of the interarrival times of different coordinates, we compute the change of measure as in~\eqref{eq:girsanov_one_dim} jointly over all components: 
\begin{align}
\nonumber & p_{Z}(\tau^{(1)}, \ldots, \tau^{(d)}) = \prod_{k=1}^d p_{e}(\tau^{(k)}; \delta) \\
\nonumber &\mathbb{E}_{\mathbb{Q}_\updownarrow}\Bigg[\exp\bigg(\int_0^{\left( \bigvee_{k=1}^d \tau^{(k)} \right) \bigwedge T} \mu({\bf e}_t,t;\theta)\mathrm{d}{\bf e}_t \\& \quad \quad - \frac12 \int_0^{\left( \bigvee_{k=1}^d \tau^{(k)} \right) \bigwedge T} \mu^{\dagger}\mu({\bf e}_t,t;\theta) \mathrm{d}t \bigg) \Bigg]
\label{eq:girsanov_nd}
\end{align}

where ${\bf e}_t$ is a multidimensional excursion process with the excursion length of the $i^\text{th}$ component being $\tau^{(i)}$.
Specifically, ${\bf e}_t$ has zeros only at the time points where the corresponding component has a realization. 

\section{Simulating Point Processes}
\begin{figure}
    \centering
    \vspace{-20pt}
    \includegraphics[width=0.23\textwidth]{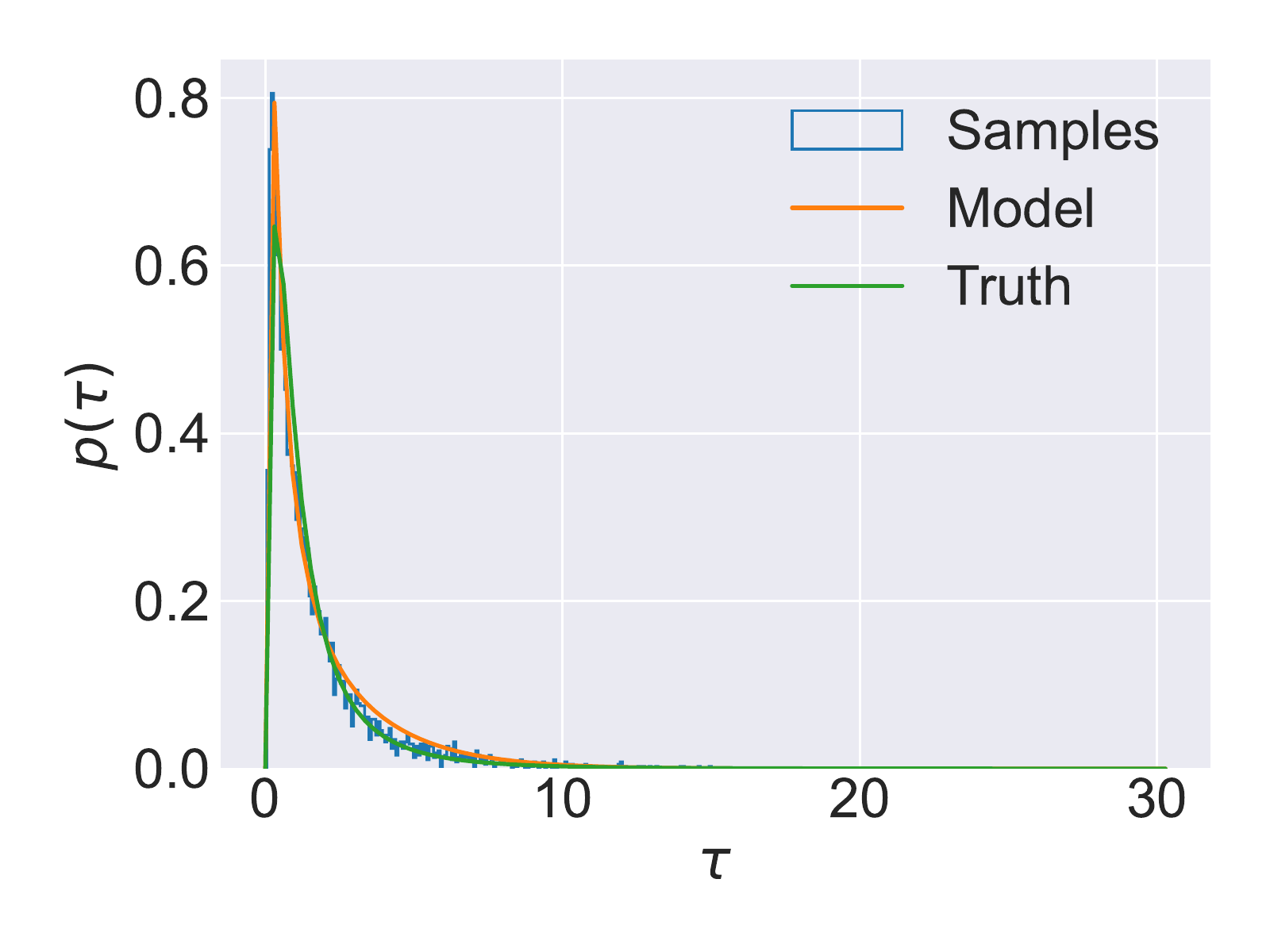}
    \includegraphics[width=0.23\textwidth]{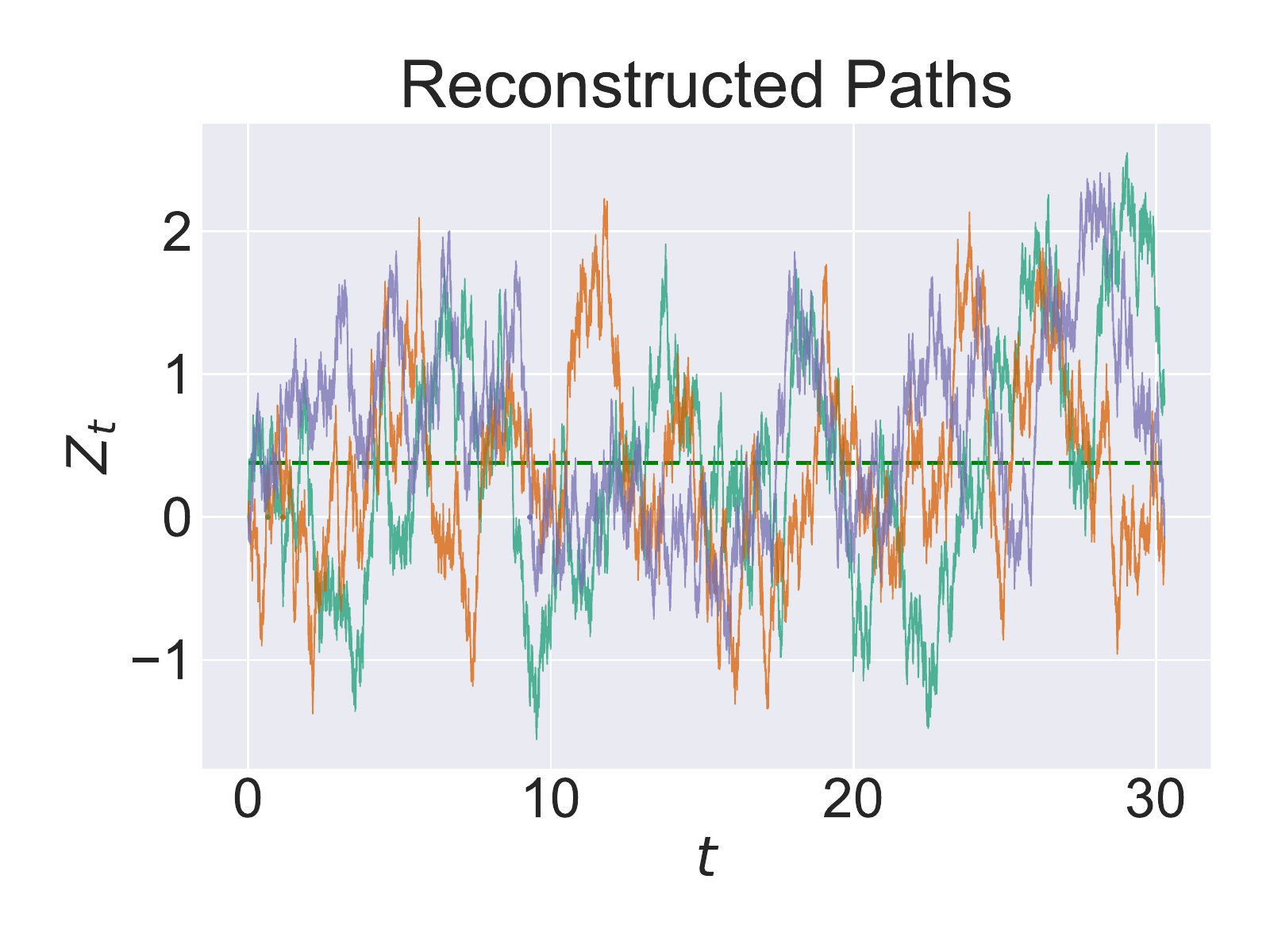}
    \caption{Example of estimated log-normal renewal process with samples generated from learned diffusion. Left: histogram of samples compared with the true density and estimated density. Right: learned sample paths with excursion lengths corresponding to the histogram. The dashed line corresponds to $\delta$.}
    \label{fig:samples}
\end{figure}

Simulating realizations of the proposed point process follows from existing simulation techniques for diffusions stopped at a boundary. 
For example, when solving certain linear elliptic PDEs, the solution is based on computing the first hitting time of diffusion on the boundary of the domain~\citep{gobet2010stopped}.
We propose a method based on the Euler-Maruyama method where excursions are simulated by computing full sample paths and finding the times when an excursion occurs. 
Importantly, this allows for both sampling the full sample path based on the fitted drift and obtaining samples of interarrival times.
We summarize the heuristic for the sampling procedure in Algorithm~\ref{alg:sampling}.
Figure~\ref{fig:samples} shows an example of estimating a log-normal distribution.
The figure on the left shows the histogram of samples in blue, the model probabilities computed using~\eqref{eq:girsanov_one_dim} in orange, and the true density for a log-normal distribution in green.
On the right, samples of different trajectories whose excursions to the blue dotted line result in interarrival times distributed according to the left.

\begin{algorithm}[!ht]
	\caption{Sampling arrival times} 
	\label{alg:sampling}
	\begin{algorithmic}[1]
	\STATE \textbf{Input:} Parameters of drift $\mu(x,t)$, step size $\Delta t$, total time $T$, initial state $X_0$, minimum height $\delta$ and variance $\sigma$
	\STATE Sample using Euler-Maruyama a sample path from $X_0$ to $X_t$ using step size $\Delta t$ and variance $\sigma \Delta t$:
	$$X_{t+1} \sim \mathcal{N}(X_t + \mu(X_t, t)\Delta t, \sigma \Delta t)$$
	\STATE Compute the set $\tau_0 = \{X_s = 0 : s \in [0,T]\}$ 
	\STATE Filter $\tau_0 \to \tau_\delta$ where $\tau_\delta = \tau_0 \setminus \{t_i \in \tau_0  : \max_{t_i < s < t_{i+1}}X_s < \delta\}$.
	\STATE \textbf{Return.} Set of arrival times $\tau_\delta$.
	\end{algorithmic} 
\end{algorithm}
The simulation algorithm for a more complicated process follows a similar procedure, for example, the drift term $\mu(x,t)$ can be replaced by the history-dependent function $\mu(t, \mathcal{H}_t)$, which we discuss further in the Appendix.

\section{Practical Considerations}
Here we discuss some practical considerations regarding the model. 
We first describe how partitioning the space of excursions can lead to a marked point process without resorting to a $d$-dimensional latent process.
We then describe a result regarding the family of interarrival distributions the method can represent,
We finally describe the relationship between the FHT problem and the excursions approach.

\subsection{Multi-dimensional Point Process from a Single One-dimensional Diffusion}
A unique property of the proposed method is the ability to represent a multi-dimensional point process with a single one-dimensional latent diffusion process. 
The main idea comes from partitioning the measure on the space of excursions to correspond to different classes. 
In the simplest case, the arrival can come from an excursion above or below the reference level.
The structure that should be maintained is a natural ordering between the marks of the point process for discrete marks. 
For example, in the case of a two-dimensional process, one set of marks should always be greater than the other. 
This can then correspond to the running maximum and running minimum times. 
Note that this differs from the $d$-dimensional process that was described in Section~\ref{sec:multi} which assumes a $d$-dimensional noise source.
We describe potential applications in a financial setting in Appendix~\ref{sec:examples} where bids and asks in an opaque market are generated by the running maximum or running minimum process.

\subsection{Expressiveness}
A relevant question asks how expressive the class of interarrival times generated by excursions is.
We characterize this in the following remark:
\begin{remark}
Consider an excursion length distribution for a fixed $\delta$ as  $p_{\updownarrow,\delta}$ with support on $\mathbb{R}_+$ and a distribution that we wish to approximate using the excursion length of an Ito diffusion as $p_\star$.
Define the function space such that the excursion distribution is a density as a subset of Lipschitz functions $\mathrm{Exc} \subset \mathrm{Lip}(\mathbb{R}_+, \mathcal{D})$.
For a fixed integration time $T$, the excursions of the diffusion with drift $\mu$ can represent $p_\star$ if the 1-Wasserstein distance between the two is less than
\begin{align*}
 \sup_{\mu \in \mathrm{Exc}(\mathbb{R}_+, \mathcal{D})}\sqrt{\frac{T^2}{2} \mathbb{E}_{X_t \sim\mathbb{ Q}}\left[\int_0^T \mu \mathrm{d} X_t - \frac12 \int_0^T \mu^2 \mathrm{d} t \right]}.
\end{align*}
\end{remark}
\begin{proof}[Intuition of Proof]
This follows bounding the Wasserstein distance using Pinsker's inequality. 
The full proof is in Appendix~\ref{sec:proofs}.
\end{proof}
The remark allows for a simple condition on whether a distribution can be approximated by the proposed method for fixed integration time and excursion height. 
The 1-Wasserstein distance is easy to calculate since it is the difference between the CDFs, making the remark useful in practice since one can certify whether a distribution can be represented using the change of measure.
More generally,~\citet[Section 6]{pitman2007ito} discusses the applicability of functions of Brownian excursions representing the full class of stable L\'evy processes at the cost of relaxing many of the assumptions on the drift of the process.

\begin{figure}
    \centering
    \includegraphics[width=0.48\textwidth]{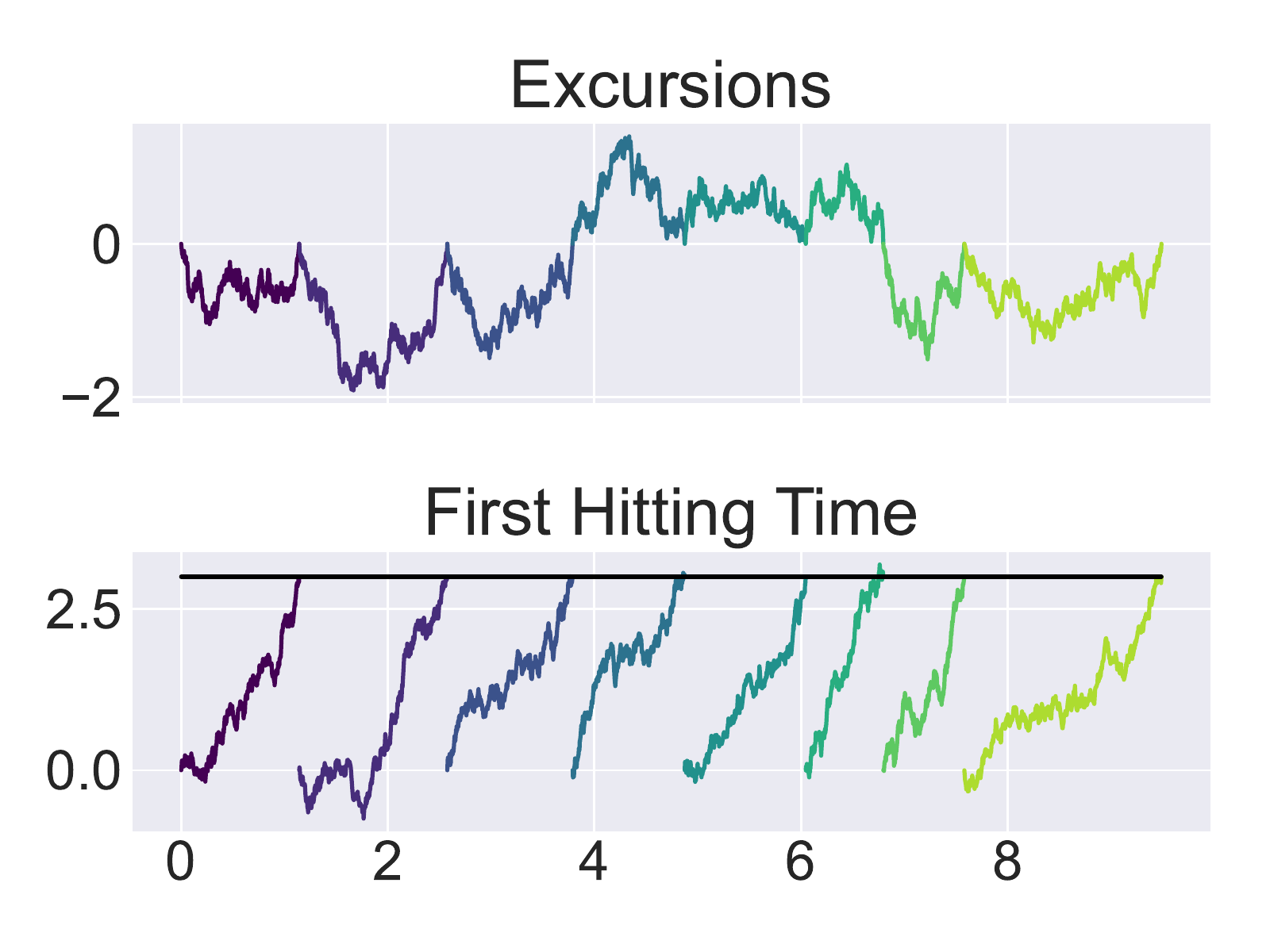}
    \caption{Schematic comparing excursions (top) with first hitting times (bottom). The excursions result in a continuous path whereas the first hitting time approach results in discontinuous paths that need to be reset after each arrival.}
    \label{fig:fht_vs_ex}
\end{figure}

\subsection{Comparing Excursions and First Hitting Times}
\label{sec:fht}

We now describe in greater detail the similarities and differences between studying excursions and FHTs related to the motivation of this work. 
Figure~\ref{fig:fht_vs_ex} provides a qualitative description of the difference between the excursion representation we consider here and the first hitting time approach. 

\begin{figure*}
    \centering
	\subfloat[10-$d$ $\mu = -\tanh(x)$;]{\includegraphics[width=0.24\textwidth]{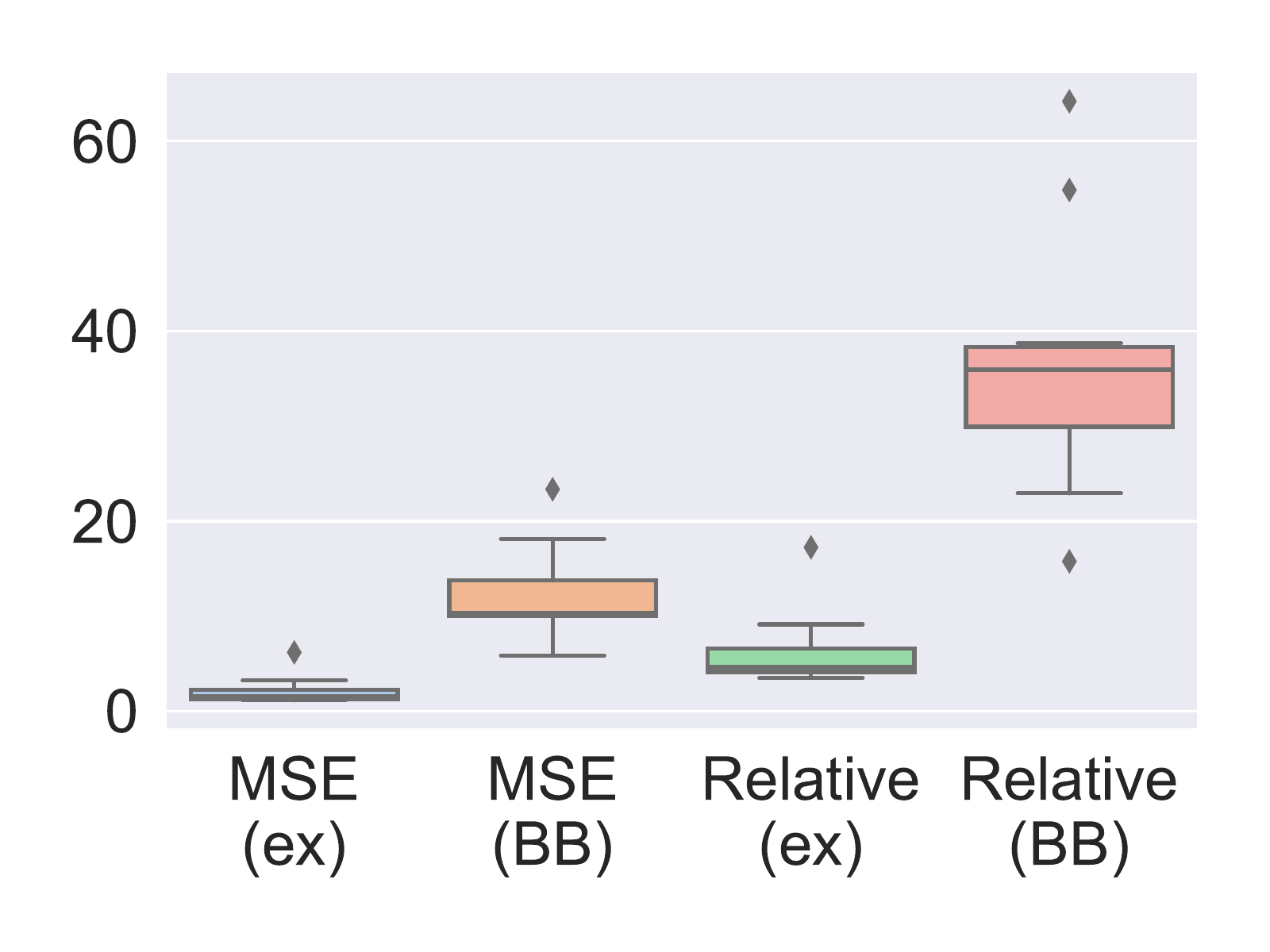}}
	\subfloat[10-$d$ $\mu = -x^3$;]{\includegraphics[width=0.24\textwidth]{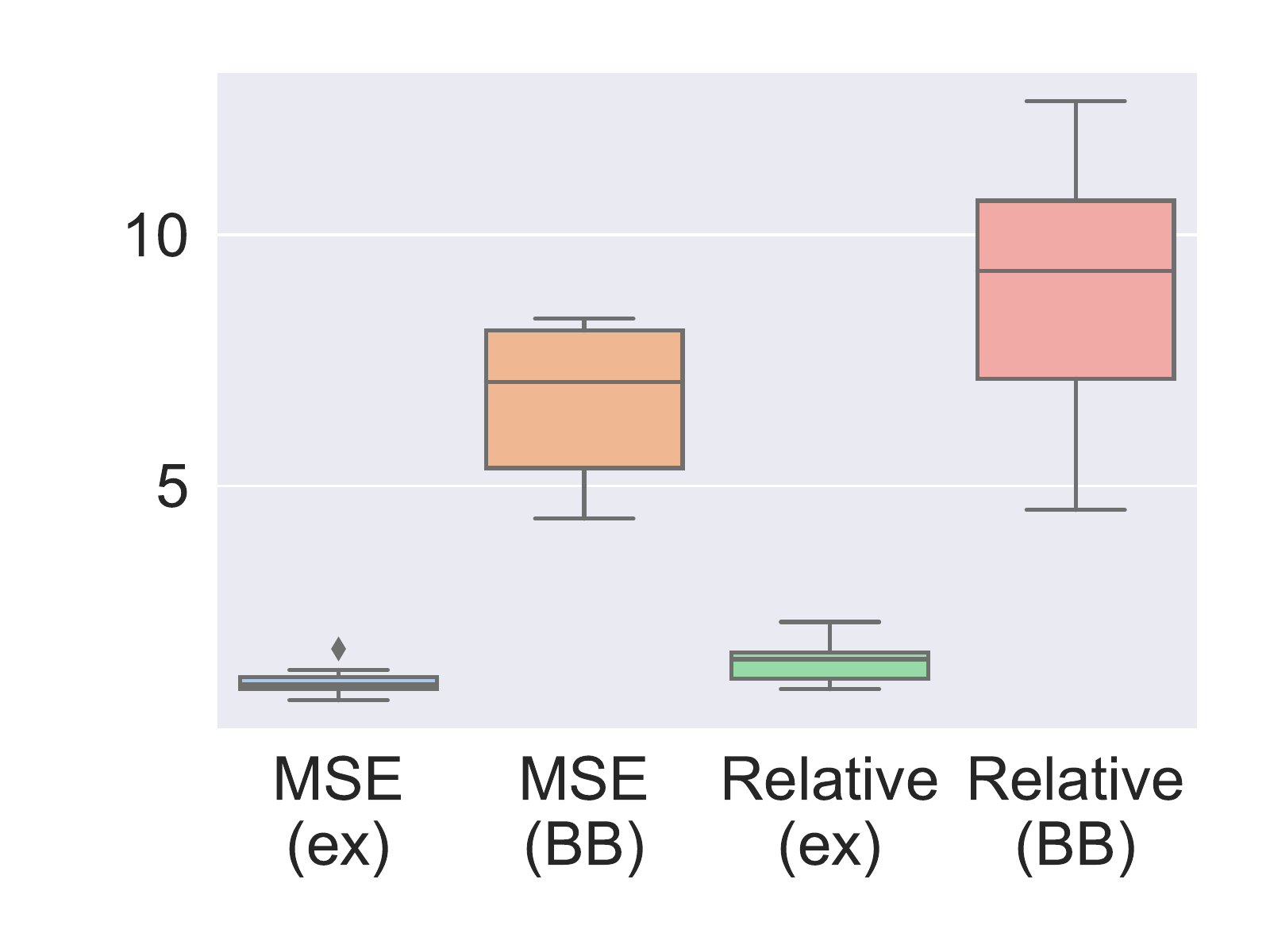}}
	\subfloat[2-$d$ $\mu = circle$;]{\includegraphics[width=0.24\textwidth]{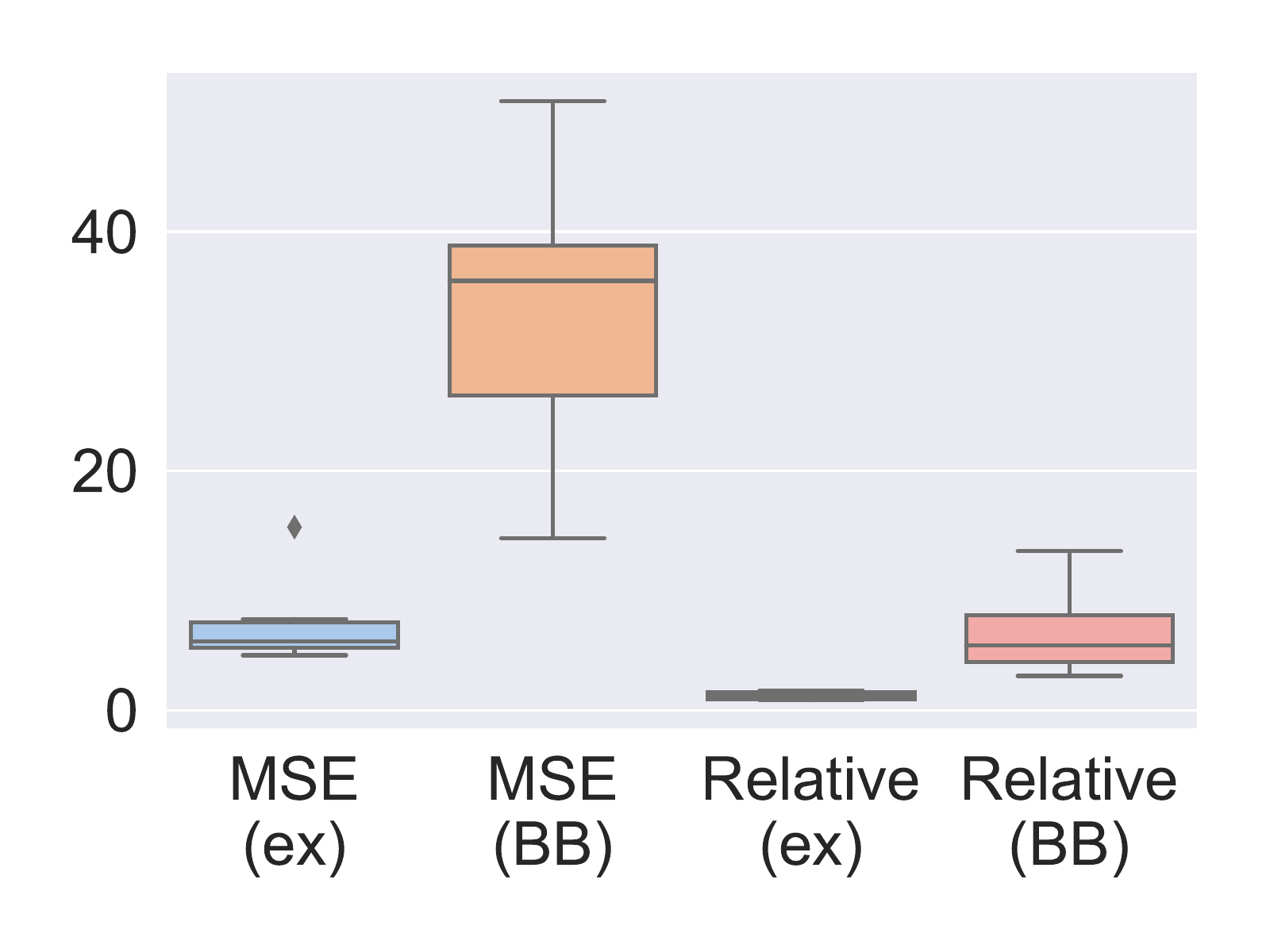}}
	\subfloat[5-$d$ $\mu = -x$;]{\includegraphics[width=0.24\textwidth]{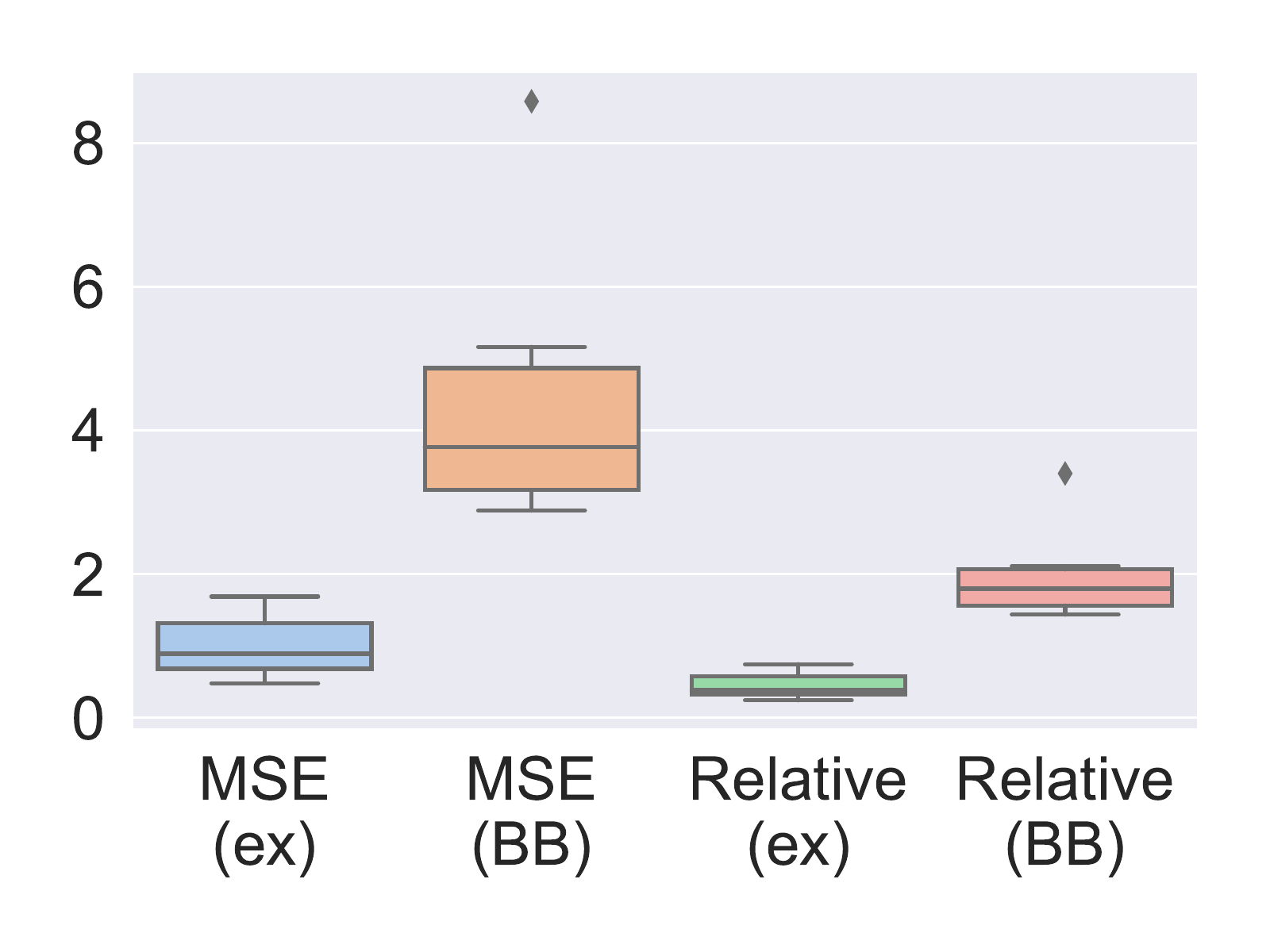}}
    \caption{Examples of point processes with interarrival distributions given by different $\mu$. Interarrival times are given as a crossing of $0$. Results from 10 runs. Relative refers to the MSE divided by the norm of the solution. BB refers to the Brownian bridge estimator and Ex refers to the proposed excursion estimator.}
    \label{fig:mse_nd}
\end{figure*}

\paragraph{Interpretations of Excursions and Hitting Times}
Excursions and hitting times share many similarities, since an excursion may be seen as the FHT to $\delta$ and back to $0$ again.
This relationship is specified in detail in Proposition~\ref{prop:laplace} found in Appendix~\ref{sec:laplace} where the Laplace transform of the first hitting time distribution is given.
In that sense, the primary reason to consider one representation versus the other is the interpretation of the underlying phenomena being observed.
The first hitting time density requires the assumption that the particle is returned to its original state at $t=0$ for every subsequent arrival.
This makes the full sample path discontinuous since the particle must hit a level $\delta \neq X_0$ and then instantly return to $X_0$.
On the other hand, considering an excursion from a level yields a continuous sample path for a full sequence of observations.
Existing literature has considered this relationship, as in~\citet[Proposition 3.4]{ananova2020excursion} where a sample path of diffusion can be reconstructed from excursions.

\paragraph{Densities Described by Laplace Transforms}

The problem of FHT and excursions are closely related, since both concern properties of diffusion as they approach different regions of the state space.
Both have been historically studied through their Laplace transforms.
For a univariate autonomous SDE, there exists a Sturm-Liouville problem associated with the Laplace transform of the excursion length distribution~\citep{yen2013local}.
The FHT density has also been studied through the same mathematical formulation.
However, working with the Laplace transform representation is difficult as it is necessary to invert the Laplace transform to obtain the density.
Inverting the Laplace transform is numerically unstable and also prone to numerical errors. We present a detailed description of the approach in Appendix~\ref{sec:laplace}.

\paragraph{Connection to the Running-Maximum and the Drawdown Processes}
Finally, one important property of excursions is the relationship between the running maximum of a Brownian motion and the zeros of a reflected Brownian motion.
Specifically, the identity
$$
\sup_{s < t} W_s - W_t \overset{d}{=} |W_t|
$$
where $W_t$ is Brownian motion. 
This allows the interpretation where excursions are related to times the process reaches its running maximum. 
For example, we can consider excursions above certain functionals of the process to provide additional physical interpretations of the excursion process as we illustrate in the next section.
This interpretation is not possible when considering only FHT where the diffusion must be reset at each arrival. 

\section{Experiments}
\label{sec:experiments}

\begin{figure*}
    \centering
    \subfloat[Exponential]{\includegraphics[width=0.23\textwidth]{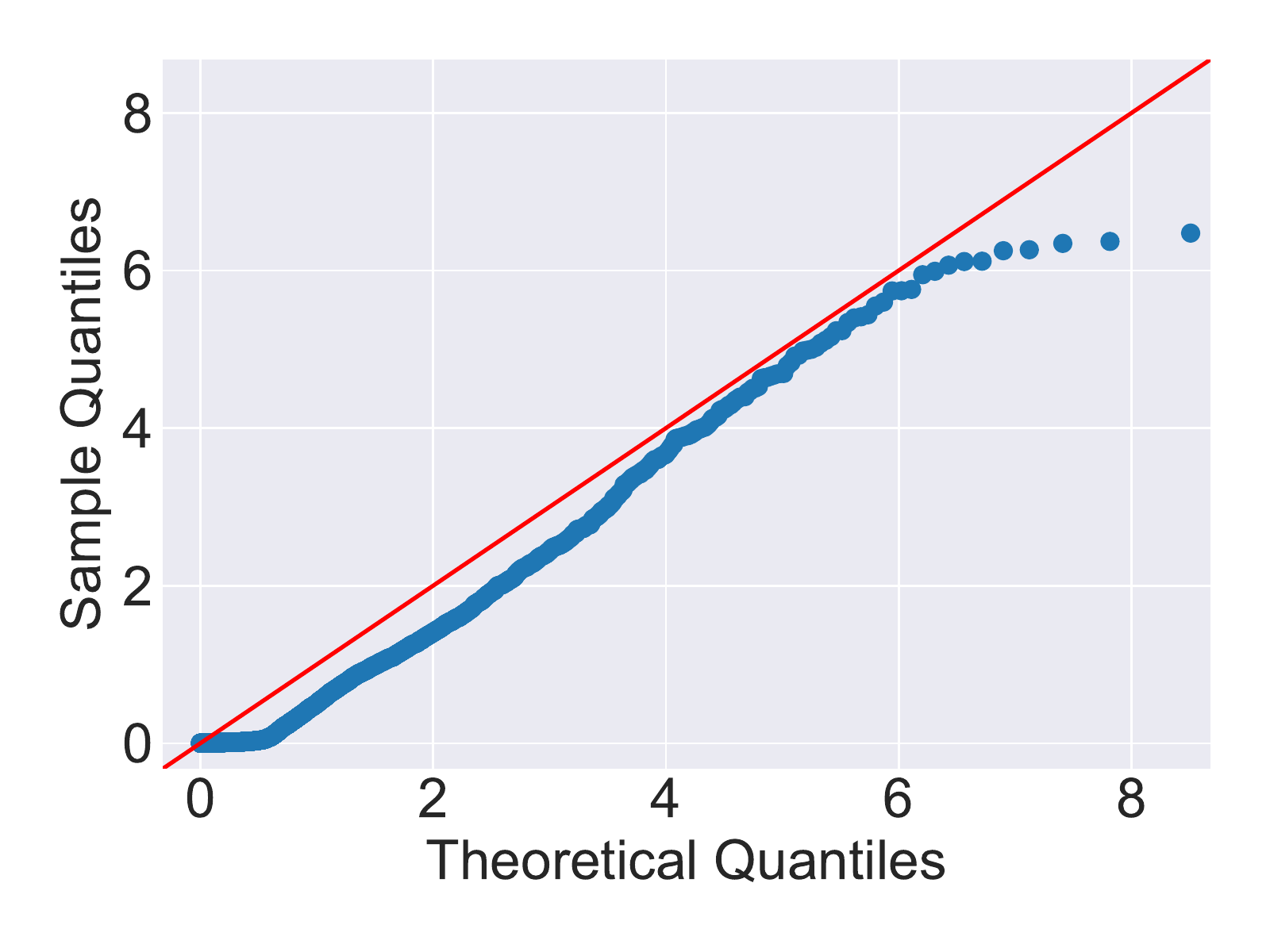}}
    \subfloat[Weibull]{\includegraphics[width=0.23\textwidth]{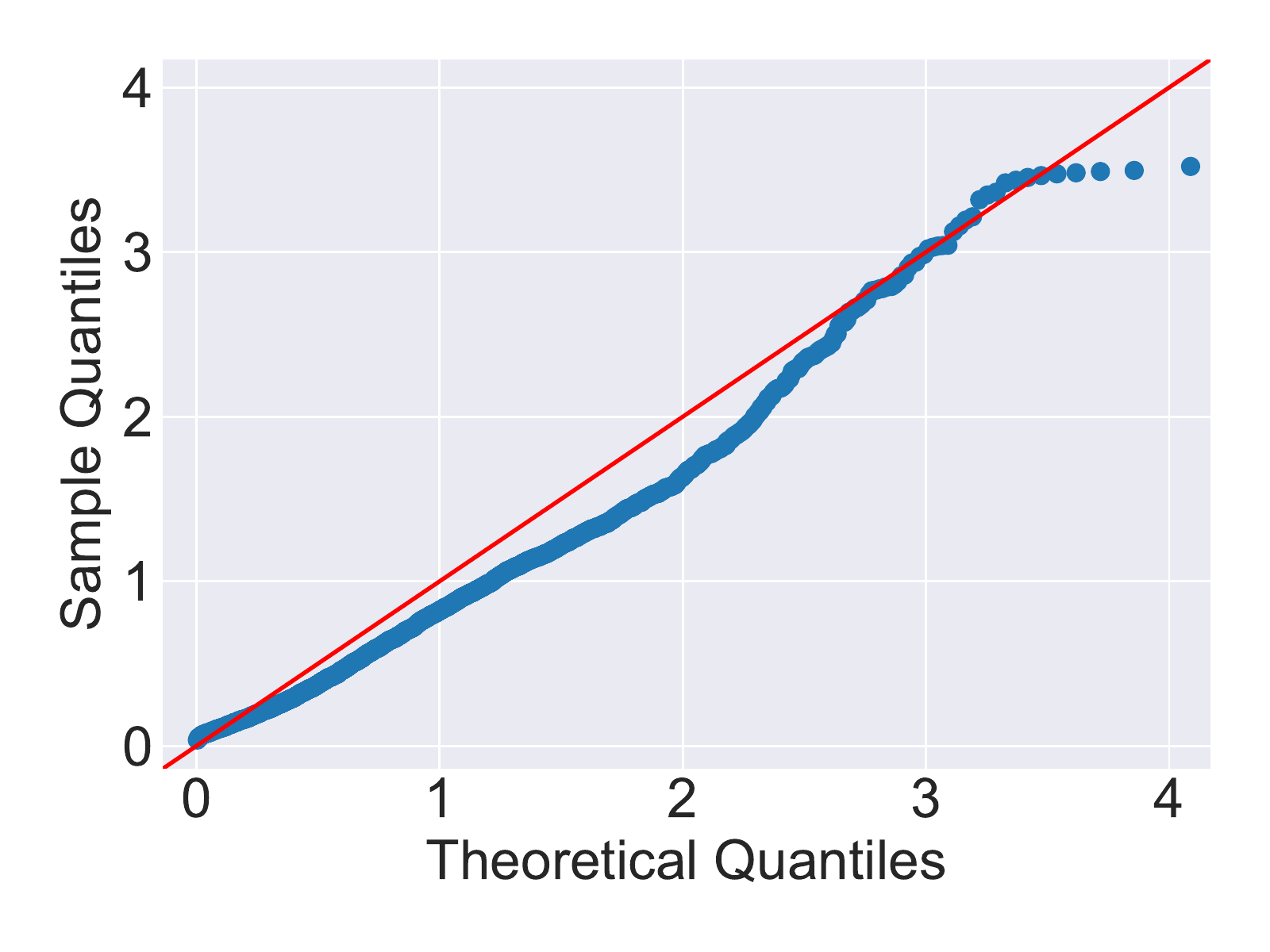}}
    \subfloat[Gamma]{\includegraphics[width=0.23\textwidth]{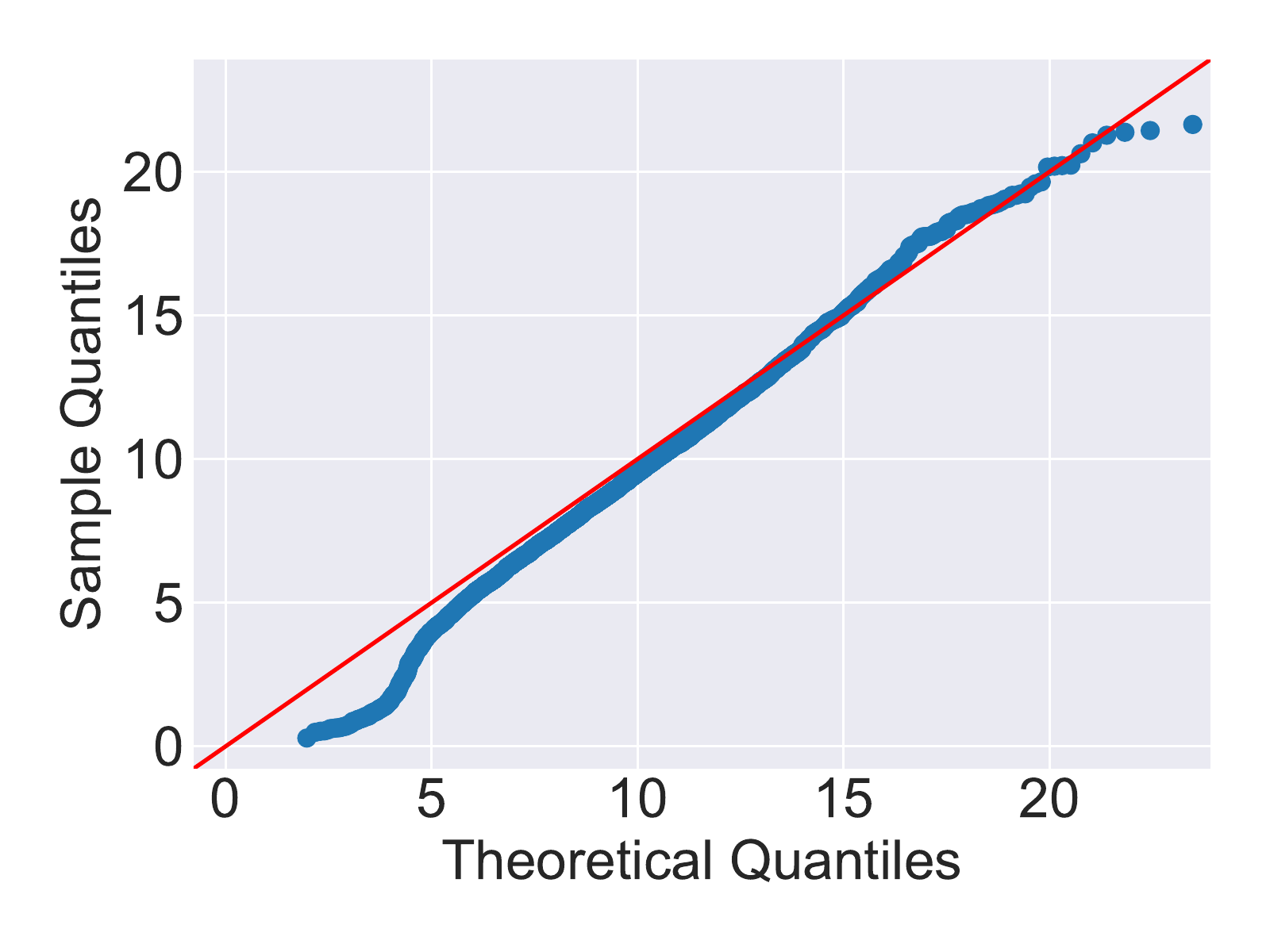}}
    \subfloat[Log-Normal]{\includegraphics[width=0.23\textwidth]{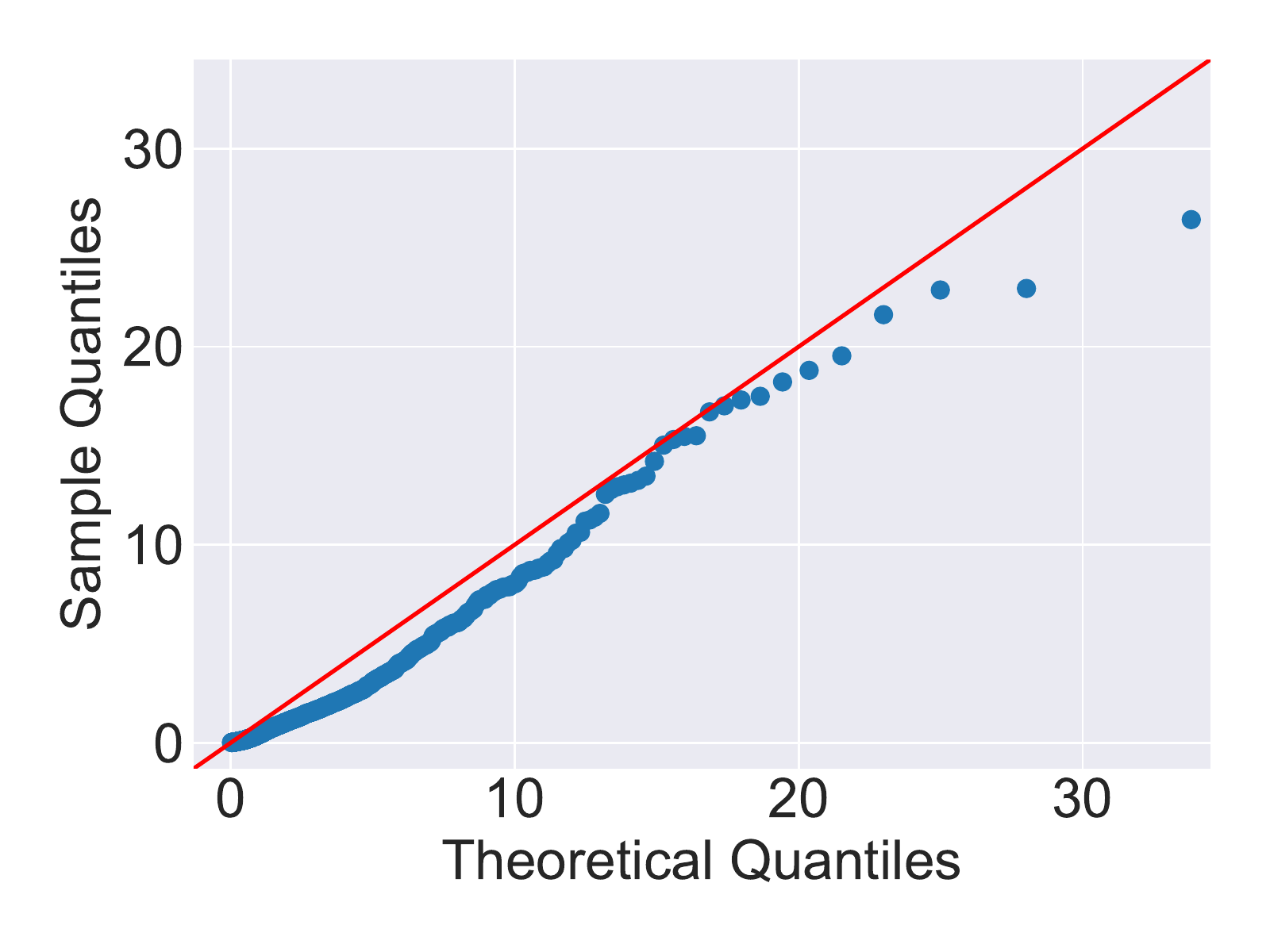}}
    \caption{Comparison of QQ plots for Poisson process, Weibull, Gamma, and log-normal renewal processes. All are fit using MLE with 200 samples from the specified renewal process.}
    \label{fig:qq}
\end{figure*}
We now consider the modeling framework in a number of synthetic and real data experiments.
The first set of experiments is based on observing the zero times of different diffusions.
These experiments examine how well the true drift can be recovered using the proposed estimator.
The baseline for these experiments is a standard SDE regression algorithm that does not consider Brownian excursions but instead considers Brownian bridges. 
The second set of experiments analyzes the proposed estimator in representing the interarrival distributions of canonical renewal processes.
Finally, we provide a real example regarding a  physical process where the underlying behavior is posited to be related to a continuous process. 
In this case, we consider how well the learned diffusion agrees with the latent factor that's known to cause the point process.

\subsection{Recovery of Drift from Excursion Lengths}

In order to validate the method in the context of a latent diffusion, we consider a series of experiments on how well the method can recover the drift of the latent diffusion.
This would correspond to a real scenario where the data are generated according to excursions of a diffusion. 
We observe $\{(\tau_i, m_i\}_{i=0}^N$ where $m \in \mathbb{N}$ is the mark corresponding to the dimension where the excursion occurs. 
We compute this for different choices of $\mu$ and compute the mean squared error (MSE) between the estimated $\hat{\mu}$ and the true $\mu$.
The observation are generated by first simulating a diffusion and finding the zero times of the sample paths.
We detail the different models in Table~\ref{tab:mult_d_exp} in the appendix.
The results are illustrated in Figure~\ref{fig:mse_nd} where we compare the estimation based on maximizing the likelihood of diffusion with unknown drift based on a path integral estimator. 
The difference between the estimators is that the proposed one uses the expectation over excursions whereas the other considers Brownian bridges.
The Brownian bridges do not use the complete information of the problem, and therefore result in higher errors as well as higher variances than the Brownian excursions.

\begin{figure*}[h!]
    \centering
\includegraphics[trim=0pt 50pt 0pt 0pt, width=\textwidth]{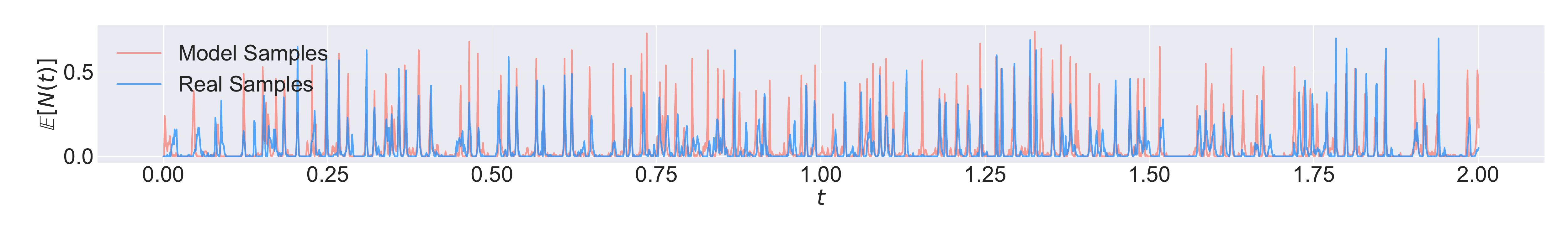}
    \caption{Comparison of the empirical histogram of arrival times versus the sampled histogram of arrival times. Sampled spikes generally align with true spikes.}
    \label{fig:hist_neuron}
\end{figure*}

\paragraph{History and Exogenous Signal Dependent Processes}
We consider an extension of the previous experiment in the 1-dimensional case where we recover the coefficient of the drift when it is dependent on either the history of hitting times or on an observed exogenous signal. 
Specifically, we define a drift of the following form:
$
\mu(X_t, \mathcal{S}_t) = -X_t + w\varphi(t - \mathcal{S}_t)
$
where $w$ is the coefficient of interest and $\varphi\left ( \cdot \right)$ is a known kernel that influences the drift based on the observed history or exogenous process. 
We choose $\varphi$ to be the exponential kernel defined as $\varphi = \exp\left(-\left(t - \mathcal{S}_t \right)/\eta\right)$ with $\eta$ a fixed parameter. 
In the case of history dependence, $\mathcal{S}_t$ is given as $\mathbb{H}_t$ while the exogenous process is given by a randomly generated signal generated using uniform increments, full details are given in Appendix~\ref{sec:experiment_details}.
We compare the squared error of the estimated value of $\hat{w}$ to the true value of $w$ as recovered by the same Brownian bridge estimator and by the proposed excursion estimator in Table~\ref{tab:hist_table}.
The results again suggest that the proposed estimator achieves better performance than the relevant baseline of regressing a SDE to the data.

\begin{table}[h]
\footnotesize
    \centering
    \begin{tabular}{@{}lllll@{}}\toprule
   &  & $w=0.5$ & $w=1$ & $w=2$ \\
   \midrule
        \multirow{2}*{History}  & BB & 0.247(0.003)  &  0.999(0.012)& 3.996(0.021) \\
          & Ex & 0.099(0.091) & 0.067(0.068) & 0.903(0.376) \\ 
         \multirow{2}*{Input} & BB & 0.248(0.003)  & 0.997(0.018) & 3.786(0.580)  \\
          & Ex & 0.148(0.090)& 0.069(0.113) & 1.301(2.012) \\ 
          \bottomrule
    \end{tabular}
    \caption{Squared error $(\hat{w} - w)^2$ of the history coefficient for the Brownian bridge estimator (BB) versus the Brownian excursion estimator (Ex). }
    \label{tab:hist_table}
\end{table}

\subsection{Estimating and Sampling Point Processes}

Next, we are interested in determining how well the proposed method can represent some canonical point processes. 
In this case, we consider the homogeneous Poisson process, a Gamma renewal process, a Weibull renewal process, and a log-normal renewal process with 40 samples with 5 realizations in each sample for a total of 200 points. 
Full parameters of the distributions are given in the appendix.
This experiment tests both inference (Algorithm~\ref{alg:mle}) and the sampling (Algorithm~\ref{alg:sampling}).
We plot a QQ plot of the samples generated by the excursion lengths versus the theoretical quantiles in Figure~\ref{fig:qq}.
The figures suggest that the estimation and sampling methods are able to capture the distributions of the point process.  

\subsection{Real Data}
Finally, we consider a neuroscience dataset where the firing of mouse neurons is recorded as a function of an external stimulus as described in~\citet{tripathy2013intermediate}.
Full details of the dataset are in Appendix~\ref{subsec:real_data}.
Our main goal for this experiment is to determine how well the model can fit this data  and more importantly determine whether we can use the estimated latent path $Z_t$ as a signal that relates the original stimulus to the observed neuron spike times. 
We illustrate these findings in Figures~\ref{fig:hist_neuron} and~\ref{fig:learned_stim} where we compare the histograms of the averaged point processes generated by the true data samples and the estimated excursion process (Figure~\ref{fig:hist_neuron}) and compare the learned stimulus to the true stimulus (Figure~\ref{fig:learned_stim}).
The learned stimulus was obtained by transforming the sampled path $Z_t$ to $\tilde{Z_t} = a\log(\mathbb{E}[Z_t]) + b$ where $a,b$ are computed according to least squares with respect to the true stimulus.
We use the $\log$ transformation since the peaks of the stimulus correspond to new arrivals whereas the zeros of the learned process correspond to new arrivals -- applying the $\log$ then transforms the zeros to peaks.
In both cases, the alignment of the spikes between the learned and the true signal is well maintained.
\begin{figure}
    \centering
\includegraphics[trim=10pt 0pt 10pt 0pt, clip, width=0.5\textwidth]{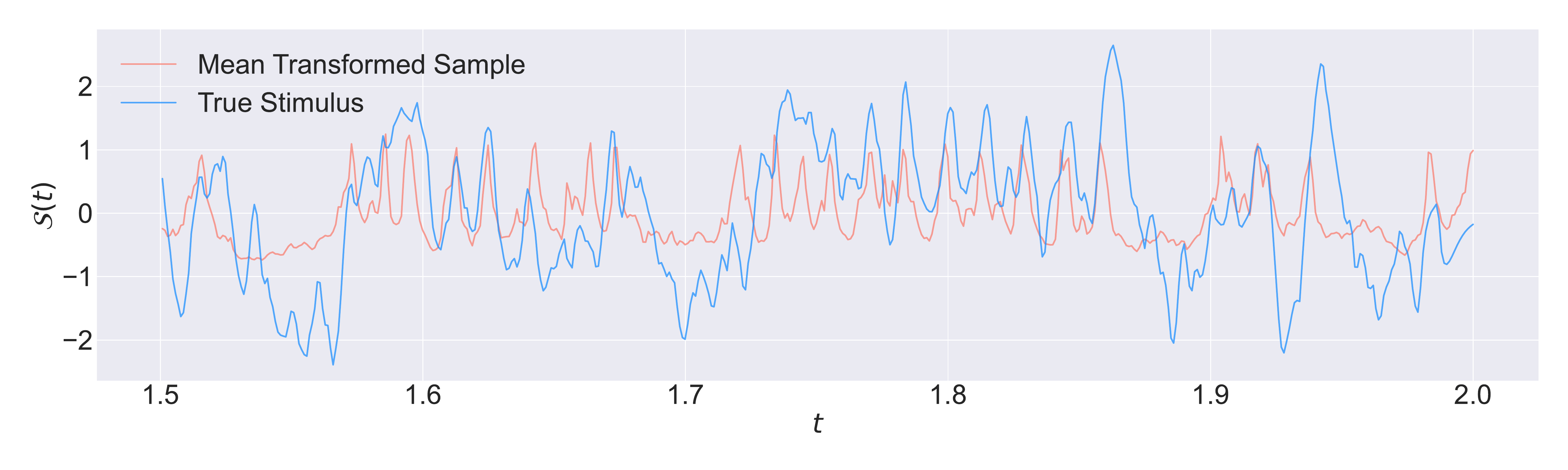}
    \caption{Average of learned sample on $t\in [1.5,2]$ paths compared with the true stimulus. Spikes of the transformed learned stimulus generally align with the true stimulus.}
    \label{fig:learned_stim}
\end{figure}

\section{Discussion}

We proposed a framework that allows for interpreting the arrivals of a point process in terms of a latent diffusion.
We described extensions to cases where the point process is multi-dimensional and depends on history or an exogenous signal.
The numerical results suggest that the estimator and the framework is useful for modeling a variety of point processes and outperforms standard SDE regression techniques.
Additionally, the results on neural data demonstrate the applicability of the proposed framework in scenarios where recovery of the unobserved continuous latent process is beneficial for analyzing a particular temporal point process. 

\paragraph{Limitations} 
The framework has a number of limitations as well. 
While we empirically validated the ability to recover the correct drift, we have no identifiability proof that guarantees the true drift will be recovered. 
In cases where the $\mu$, $\delta$, and $\sigma$ are all parameters, for example, identifiability does not hold, and restricting to a smaller class of parameters is necessary. 
Proving consistency of the estimator could be considered for follow-up work.
On a practical front, there are many situations where having an interpretation in terms of a continuous process is not appropriate. 
In these cases, more traditional point process models, such as those that rely on intensities, should be considered for the modeling task. 

\section*{Acknowledgement}
We appreciate the great support from professor Nathan Urban by sharing the valuable neuroscience dataset.
We thank Zaeem Burq and Kashif Rasul for invaluable discussions regarding the theory of excursions. 
We additionally thank Jessica Loo, Wei Deng, Volodymyr Volchenko, and Yuriy Nevmyvaka for helpful feedback on the manuscript.
AH and VT were supported in part by the Air Force Office of Scientific Research under award number FA9550-20-1-0397.
AH was also partially supported by NSF GRFP.

\bibliography{refs}
\newpage
\appendix
\onecolumn
\title{Inference and Sampling of Point Processes from Diffusion Excursions \\(Supplementary Material)}
\maketitle
We provide some technical results that supplement the main results in the text, 
 including the calculation of the excursion length density, discussion about the connections between diffusion-based and intensity-based frameworks, and Lamperti transformation for formulating a constant variance diffusion process.

\section{Proofs}
\label{sec:proofs}

\subsection{Proof of Excursion Length Density via Change of Measure}
\label{sec:proof_com}
We restate the main proposition here:
\begin{proposition}[Diffusion Excursion Density]
Let $Z_t$ satisfy an SDE with drift $\mu$ such that $Z_t$ is recurrent at zero.
Then the density of the excursion lengths $\tau$ of $Z_t$ is given by:
\begin{equation}
p_{Z}(\tau) = p_{e}(\tau; \delta)\mathbb{E}_{\mathbb{Q}_\updownarrow^\delta}\left[\exp\left(\int_0^\tau \mu(e_t,t;\theta)\mathrm{d}e_t - \frac12 \int_0^\tau \mu^2(e_t,t;\theta) \mathrm{d}t \right) \right].
\end{equation}

\end{proposition}

\begin{proof}
The proof follows a straightforward change of measure argument.
Performing the change of measure, we can write the Radon Nikodym derivative as
$\frac{ \mathrm{d} \mathbb{P}^{\mu}}{\mathrm{d} \mathbb{Q}}$ where $\mathbb{Q}_\updownarrow$ is the excursion length density;
$\mathbb{Q}_\updownarrow^{\delta}$ is the corresponding length density of $\delta$-excursions. 
Under Novikov's condition, we can write the exponential martingale as
\begin{equation}
\frac{ \mathrm{d} \mathbb{P}^{\mu}}{\mathrm{d} \mathbb{Q}_\updownarrow^{\delta}} = \exp \left ( \int_0^{\tau\bigwedge T} \mu(e_t,t) \mathrm{d}e_t - \frac12 \int_0^{\tau\bigwedge T} \mu(e_t, t)^2 \mathrm{d} t \right ).
\label{eq:radon}
\end{equation}
\noindent Now we have that 
\begin{align}
\nonumber \mathbb{P}^{\mu, \delta}(\tau \in \mathrm{d} t; \delta) &= \mathbb{E}_{\mathbb{Q}_\updownarrow^\delta}\left[ \frac{\mathrm{d} \mathbb{P}^\mu}{\mathrm{d} \mathbb{Q}_\updownarrow^{\delta}}  \mathds{1}(\tau \in \mathrm{d} t)\right] \\
\nonumber &= \mathbb{E}_{\mathbb{Q}_\updownarrow^{\delta}}\left[\exp \left ( \int_0^{\tau\bigwedge T} \mu(e_t,t) \mathrm{d}e_t - \frac12 \int_0^{\tau\bigwedge T} \mu(e_t,t)^2 \mathrm{d} t \right ) \mathds{1}(\tau \in \mathrm{d} t)\right] \\
&= \mathbb{E}_{\mathbb{Q}_\updownarrow^{\delta}}\left[\exp \left ( \int_0^{\tau\bigwedge T} \mu(e_t,t) \mathrm{d}e_t - \frac12 \int_0^{\tau\bigwedge T} \mu(e_t,t)^2 \mathrm{d} t \right ) \right] p_e(\tau; \delta)
\label{eq:com}
\end{align}
with the  $e_t$ being excursions of length $\tau$.
This results in the final PDF of $\tau$ as
\begin{equation}
p_{Z}(\tau; \delta) = p_{e}(\tau; \delta)\mathbb{E}_{\mathbb{Q}_\updownarrow^\delta}\left[\exp\left(\int_0^\tau \mu(e_t,t;\theta)\mathrm{d}e_t - \frac12 \int_0^\tau \mu^2(e_t,t;\theta) \mathrm{d}t \right) \right].
\end{equation}

\end{proof}

\subsection{Proof of Approximation Distance}
\label{sec:proof_approx}
\begin{proof}
First, by~\citet[Theorem 4]{gibbs2002choosing} note that the $1$-Wasserstein distance between measures with support on $[0, T]$ is bounded by the TV distance according to
$$
W_{1}(p_\star, p_0) \leq T \frac12 \| p_\star - p_0 \|.
$$
Now using Pinsker's inequality, we can bound the TV distance in terms of the log of the Radon-Nikodym derivative.
$$
T \| p_\star - p_0 \| \leq T \sqrt{\frac12 \int d p_0 \left[ \int_0^T \mu dX_t - \frac12 \int_0^T \mu^2 dt \right]}.
$$
Taking the supremum over all functions within the function space that satisfies the positive recurrence condition leads to the result. 
\end{proof}

\section{Details of Excursion Length Density Given by Laplace Transform}
\label{sec:laplace}
State-dependent excursion densities can be calculated by inverting the Laplace transform of the solution of a differential equation.
In practice, this is difficult since the inverse Laplace transform is numerically unstable and intractable in general.
We do not follow this approach in the paper but state the details of the approach here for completeness. We define the procedure in the following proposition based on~\cite{pitman1999laplace}:
\begin{proposition}[\cite{pitman1999laplace}]
\label{prop:laplace}
Let $Z_t$ be the solution to the SDE in~\eqref{eq:sde}. 
Then define the speed and scale measures
\begin{equation}
s'(x) = \exp \left( -\int_0^x 2 \mu(u)/\sigma(u) \mathrm{d}u \right), \quad m(x) = \frac{2}{\sigma(x) s'(x)}.
\label{eq:speed_scale}
\end{equation}
Additionally, define the operator $\mathcal{A}$ as
$$
\mathcal{A} := \frac{1}{m(x)}\frac{\mathrm{d}}{\mathrm{d}x}\frac{1}{s'(x)}\frac{\mathrm{d}}{\mathrm{d}x}.
$$
Then the Laplace transform of the hitting time distribution is given by solving the eigenvalue problem of 
$$
\mathcal{A} \phi(x) = \lambda \phi(x)
$$
for the function $\phi(x)$ with Wiener-Hopf factorization given by $\phi_\pm(x)$: 
\begin{equation}
\mathbb{E}_x\left [e^{\left(-\lambda \tau\left(Z_t\right)\right)} \right] = \left \{ \begin{array}{cc}
    \phi_{\lambda,-}(x_0)/ \phi_{\lambda,-}(x_1) & x_0 < x_1\\
     \phi_{\lambda,+}(x_0)/ \phi_{\lambda,+}(x_1) & x_0 > x_1
\end{array} \right .
\label{eq:laplace_FHT}
\end{equation}
Then, using the strong Markov property, the total duration of the excursion from $x_0$ to $x_1$ and $x_1$ to $x_0$ (assuming $x_1 > x_0$) is given by

\begin{equation}
\mathbb{E}\left [e^{\left(-\lambda \left ( \tau \right ) \left(Z_t\right)\right) }\right] = \frac{\phi_{\lambda, -}(x_0)}{\phi_{\lambda, -}(x_1)}\frac{\phi_{\lambda, +}(x_1)}{\phi_{\lambda, +}(x_0)}.
\label{eq:laplace_ex}
\end{equation}
\end{proposition}
Notably, the inversion of the Laplace transform numerically is unstable, leading to difficulties in recovering the density from an optimized $\mu$ using this technique.

As the drift and the variance terms in the SDE defined above are only state-dependent, it is not straightforward to depict the general Ito process as in Eq \ref{eq:sde}.
For the forward problem, a remedy is: first, start with deriving the first-hitting-time property of a simple SDE with only state-dependent drift and variance such as Brownian motion, linear drift Brownian motion, Bessel process, OU process \citep[Part II]{borodin1997handbook}, so that the eigenvalues and eigenfunctions can be solved easily with closed-form; then apply the Girsanov's theorem to extend the property for more general SDE. 
But for the backward problem, the Laplace transformation of the hitting time density does not usually have closed-form except for a few cases where the Laplace transform of the first-hitting-time density can be calculated in closed-form and can be decomposed into increasing eigenfunction and decreasing eigenfunction as in \eqref{eq:laplace_FHT}, solving SDE drift or variance is thus not straightforward.

Here we present one example of solving the drift function through the Laplace transform for the first-hitting-time problem (not yet the excursion problem).
The example illustrates the difficulty of formulating the diffusion process via the Laplace transform from the modeling perspective as opposed to the proposed likelihood strategy in the main text.
We use the OU process for the forward problem, which is already well-known. 
Next, we focus on the backward problem, i.e. recovering the drift of an SDE given the first-hitting-time density.
Let the process start from $x_0$ and the hitting boundary be a constant $\alpha$.
The target first-hitting-time density is
\begin{equation}
p(t) = \frac{2|\alpha-x_0|}{\sqrt{2\pi}} (e^{2t}-1)^{-\frac{3}{2}}
    \exp\left\{ 2t - \frac{(\alpha-x_0)^2}{2 e^{2t} - 2} \right\}.
\end{equation}
Next, we solve the corresponding SDE by setting $a(x) = 1$ with constant variance.
The Laplace transform of the density is,
\begin{equation}
\mathbb{E}_{x_0}[\exp(-\lambda H_\alpha) ]
= \begin{cases}
\frac{\exp(x_0^2) D_{-\lambda}(- \sqrt{2}x_0) }
    { \exp(\alpha^2) D_{-\lambda}(- \sqrt{2} \alpha) }
= \frac{\Phi_{-,\lambda}(x_0 )}{\Phi_{-,\lambda}(\alpha )}, 
    \quad \alpha > x_0 \\
\frac{\exp(x_0^2) D_{-\lambda}(\sqrt{2}x_0) }
    { \exp(\alpha^2) D_{-\lambda}(\sqrt{2} \alpha) }
= \frac{\Phi_{+,\lambda}( x_0 )}{\Phi_{+,\lambda}( \alpha )}, 
    \quad x_0 > \alpha \\
\end{cases}
\label{eq:OU_fpt_laplace}
\end{equation}
where
\begin{align*}
D_{-\lambda}(x) :=& \frac{1}{\Gamma(\lambda)} e^{-\frac{x^2}{4}}
    \int_0^\infty t^{\lambda-1} e^{-x t - \frac{t^2}{2}} \mathrm{d}t
\end{align*}
which has the following properties
\begin{align*}
\frac{\mathrm{d} D_{-\lambda}(x) }{\mathrm{d} x}
=& -\frac{x}{2} D_{-\lambda}(x) - \lambda D_{-\lambda-1}(x)
    = \frac{x}{2} D_{-\lambda}(x) - D_{-\lambda+1}(x) \\
\lambda D_{-\lambda-1}(x) 
=& - x D_{-\lambda}(x) + D_{-\lambda+1}(x) \\
\frac{\mathrm{d}}{\mathrm{d} x} e^{x^2} D_{-\lambda}( \sqrt{2} x)
=& 2x e^{x^2} D_{-\lambda}( \sqrt{2} x) 
    + e^{x^2} D_{-\lambda}'( \sqrt{2} x) \\
=& 2x e^{x^2} D_{-\lambda}( \sqrt{2} x) 
    + \sqrt{2} e^{x^2} \left(- \frac{\sqrt{2} x}{2} D_{-\lambda}(\sqrt{2}x) 
    - \lambda D_{-\lambda-1}(\sqrt{2}x)  \right) \\
=& - \sqrt{2} \lambda e^{x^2}  D_{-\lambda-1}(\sqrt{2}x) \\
\frac{\mathrm{d}^2}{\mathrm{d} x^2} e^{x^2} D_{-\lambda}( \sqrt{2} x)
=& - \sqrt{2} \lambda \frac{\mathrm{d}}{\mathrm{d} x} e^{x^2}  D_{-\lambda-1}(\sqrt{2}x)
= 2 \lambda (\lambda + 1) e^{x^2} D_{-\lambda-2}(\sqrt{2}x)
\end{align*}

According to Eq \eqref{eq:laplace_FHT} and \eqref{eq:OU_fpt_laplace}, the corresponding function can be defined as,
\begin{equation}
\Phi_{+,\lambda}(x) := e^{x^2} D_{-\lambda}( \sqrt{2} x).
\end{equation}
As defined above, $D_{-\lambda}(x) > 0$, and the derivative of $\Phi_{+,\lambda}(x)$ is negative, so the function is decreasing satisfying the Laplace transformation decomposition.
Taking this back to the infinitesimal generator eigenfunction in Eq \ref{eq:speed_scale}, we have the equality below for $\lambda \neq 0$
\begin{gather*}
\lambda (\lambda + 1) e^{x^2} D_{-\lambda-2}(\sqrt{2}x) 
- b(x) \sqrt{2} \lambda e^{x^2}  D_{-\lambda-1}(\sqrt{2}x)
= \lambda e^{x^2} D_{-\lambda}( \sqrt{2} x)  \\
\iff \quad 
- \sqrt{2}x D_{-\lambda-1}(\sqrt{2}x) + D_{-\lambda}(\sqrt{2}x)
- b(x) \sqrt{2} D_{-\lambda-1}(\sqrt{2}x)
=  D_{-\lambda}( \sqrt{2} x).
\end{gather*}
So the solution is $b(x) = -x$, and the SDE satisfies,
\begin{equation}
d Z_t = -Z_t d t + d W_t
\end{equation}
which is a standard OU. This is a simple example that the Laplace transform can be derived with closed-form and the solution is achievable and is only state-dependent. In general, such convenience is not guaranteed and there may not have a state-dependent solution.

\section{Experimental Details}
\label{sec:experiment_details}
\subsection{Drift Recovery}
The list of drifts $\mu$ for the multi-dimensional drift recovery experiments is provided in Table~\ref{tab:mult_d_exp}.
We first simulate the $d$-dimensional diffusion process then compute the times when each dimension has an excursion from the corresponding axis. 
The parameter $\delta$ is chosen to be zero, but due to discretization in simulation, an infinite number of excursions is not observed. 
We use $\Delta_t = 0.01$ and simulate until terminal time $T=10$.
The time points where the diffusion crosses the axis are recorded. 
For the excursion estimator, we consider both positive and negative excursions which we simulate using a Bernoulli random variable with a probability of $0.5$.
We compare this to the bridge estimator that interpolates between zeros using Brownian bridges. 
The hyperparameters for both models are equal to ensure a fair comparison.
We train for 200 epochs, with a learning rate of $1 \times 10^{-3}$ using the Adam optimizer.
The architecture is a 64 width, 6 depth multi-layer perceptron with \texttt{Softplus} activation function.

\paragraph{Kernel Coefficient Recovery for Exogenous and History Dependent Drift}
For these experiments, we consider an OU process influenced by either an exogenous process or the history by defining the drift as $ \mu = -x + w\varphi(\,\cdot\,)$ with $\varphi(z) = \exp( - z / \sigma)$. 
We choose $\sigma = 1$ for the exogenous dependence and $\sigma = 2$ for the history dependence.
For the both experiments, we simulate using Euler-Maruyama with $\Delta t=0.05$ up to $T = 50$ and observe the zero times.
The exogenous signal is generated according to 100 samples from a uniform random variable as $\mathcal{U}(0, 40)$ and sorted by the value.
For these experiments, we only optimize over the estimated parameter $\hat{w}$ due to the lack of identifiability if both components are unknown. 
We use the same parameters and only compare the standard SDE regression using Brownian bridges and the excursion based approach.
The models are trained using the \texttt{AdamW} optimizer with a learning rate of $1.0 \times 10^{-2}$.
Figure~\ref{fig:history_path} illustrates a sample path with dependence on the times when it reaches zero.

\begin{table}[]
    \centering
    \begin{tabular}{lll}
    Experiment & $\mu(X_t) $  \\ \toprule
       Cubic  &  $ = -X_t^3$ & $ X_t \in \mathbb{R}^{10}$ \\
       Circle & $ = \begin{cases}
       -X_{1,t} - X_{2,t} \\
       -X_{2,t} + 5 X_{1,t}
       \end{cases}$  & $ X_t \in \mathbb{R}^2$ \\
       Tanh & $= -\tanh(X_t) $ & $  X_t \in \mathbb{R}^{10}$ \\
       OU & $ = -X_t $ & $  X_t \in \mathbb{R}^5$
    \end{tabular}
    \caption{Table of drifts for multi-dimensional experiments. }
    \label{tab:mult_d_exp}
\end{table}

\subsection{Renewal Processes}
\begin{figure*}[h]
    \centering
    \includegraphics[width=0.23\textwidth]{figs/sampling/qq-exp.pdf}
    \includegraphics[width=0.23\textwidth]{figs/sampling/qq-weibull.pdf}
    \includegraphics[width=0.23\textwidth]{figs/sampling/qq-gamma.pdf}
    \includegraphics[width=0.23\textwidth]{figs/sampling/qq-lognormal.pdf}\hfill
    
    \includegraphics[width=0.23\textwidth]{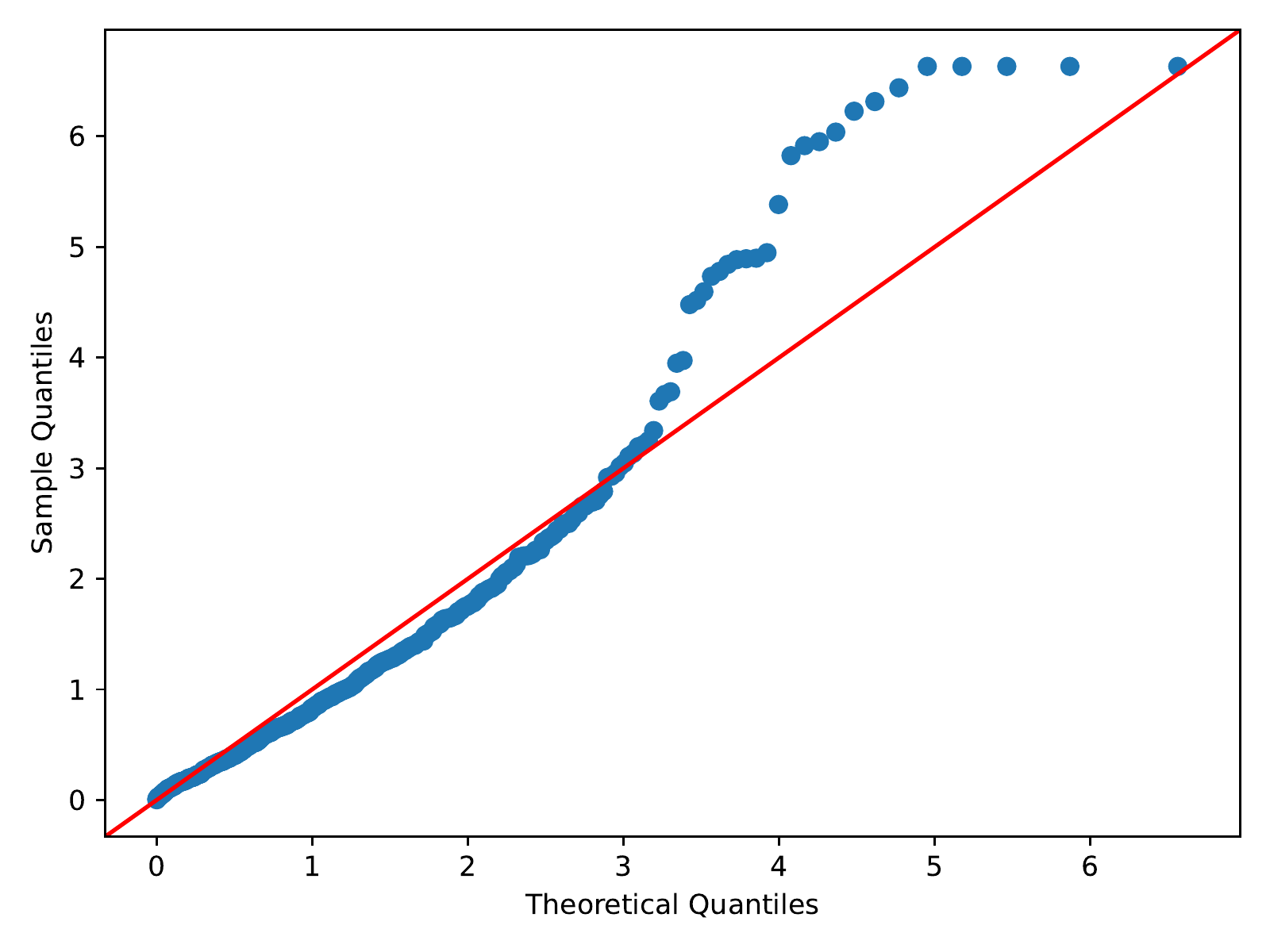}
    \includegraphics[width=0.23\textwidth]{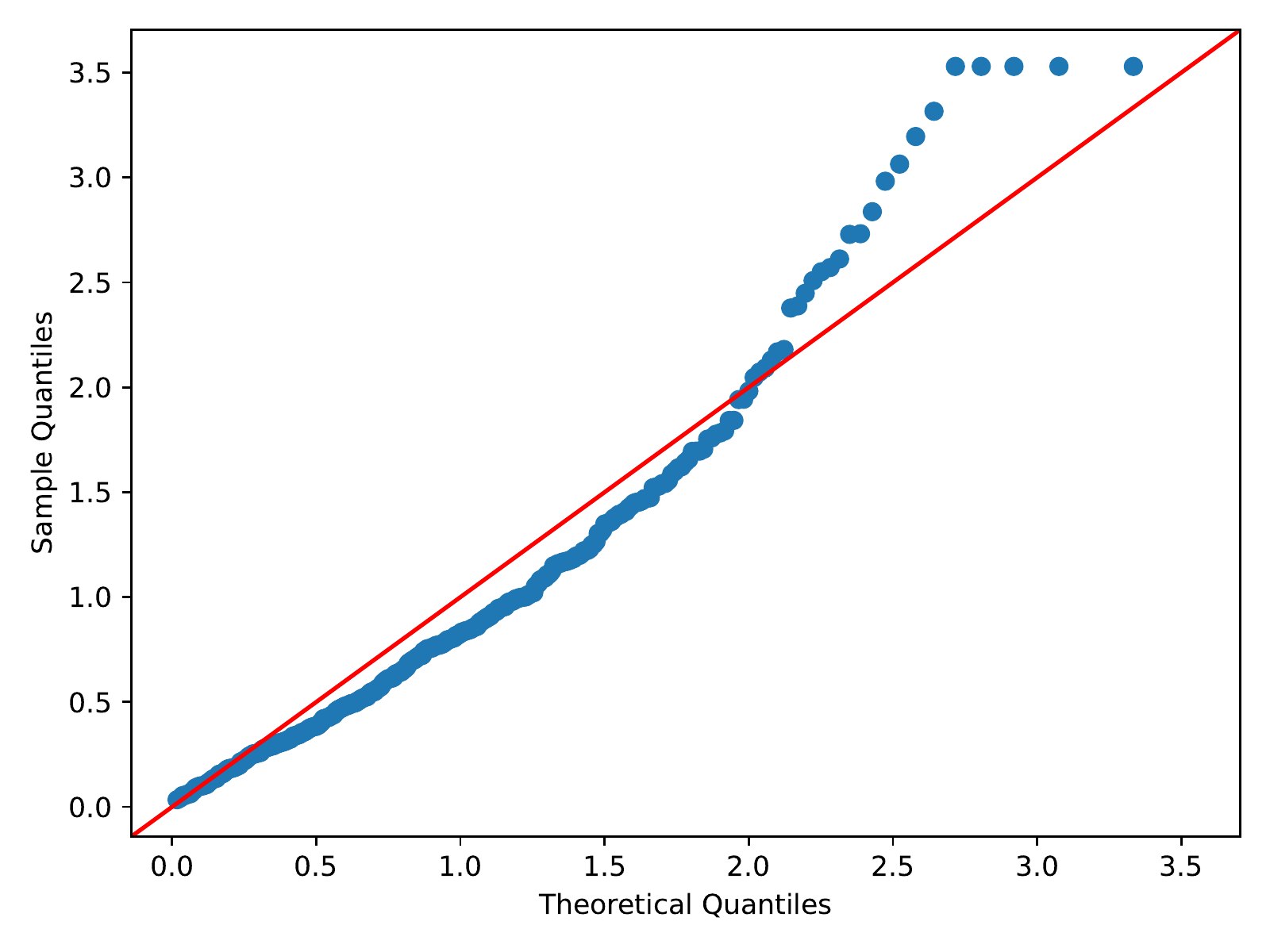}
    \includegraphics[width=0.23\textwidth]{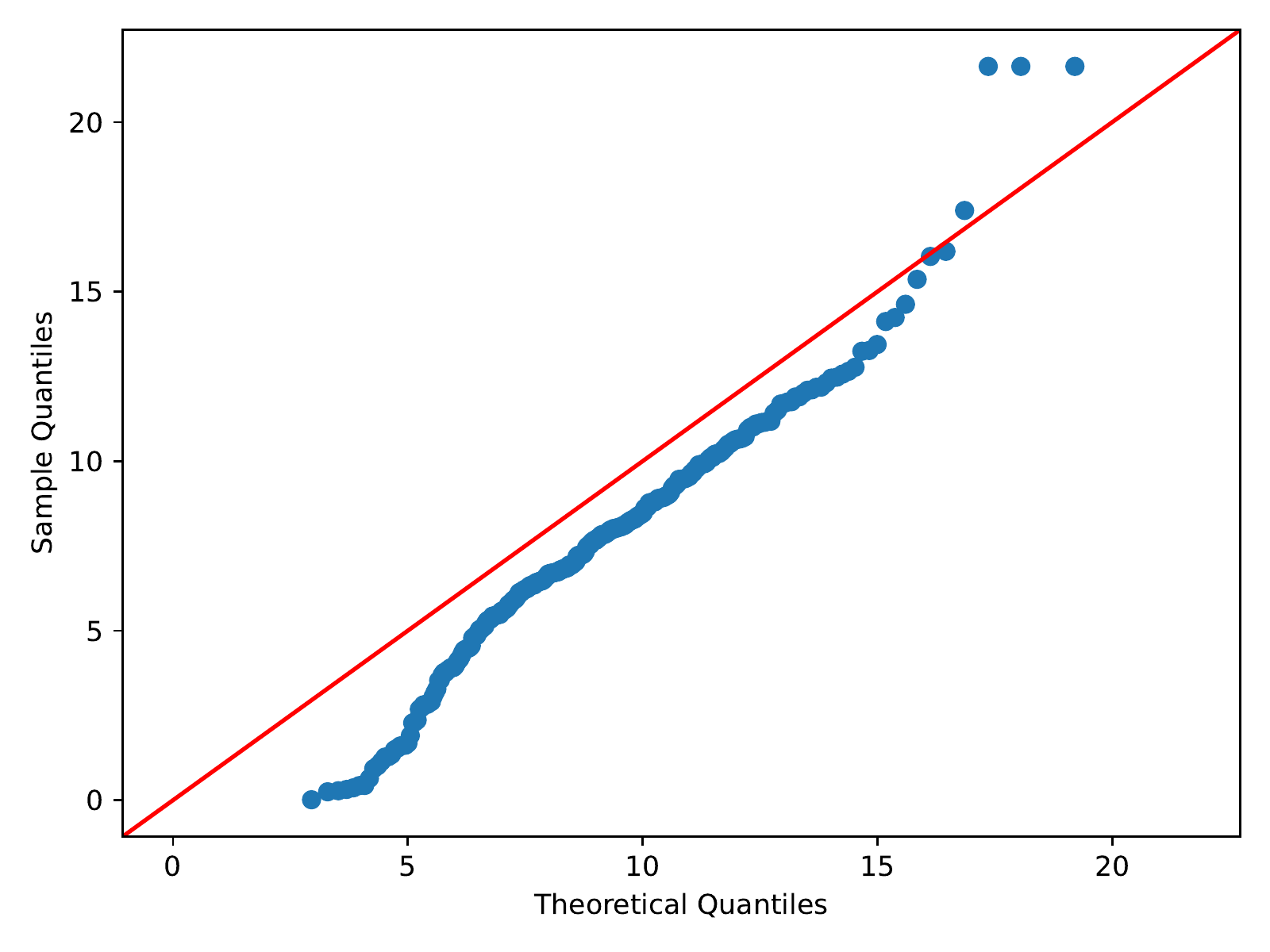}
    \includegraphics[width=0.23\textwidth]{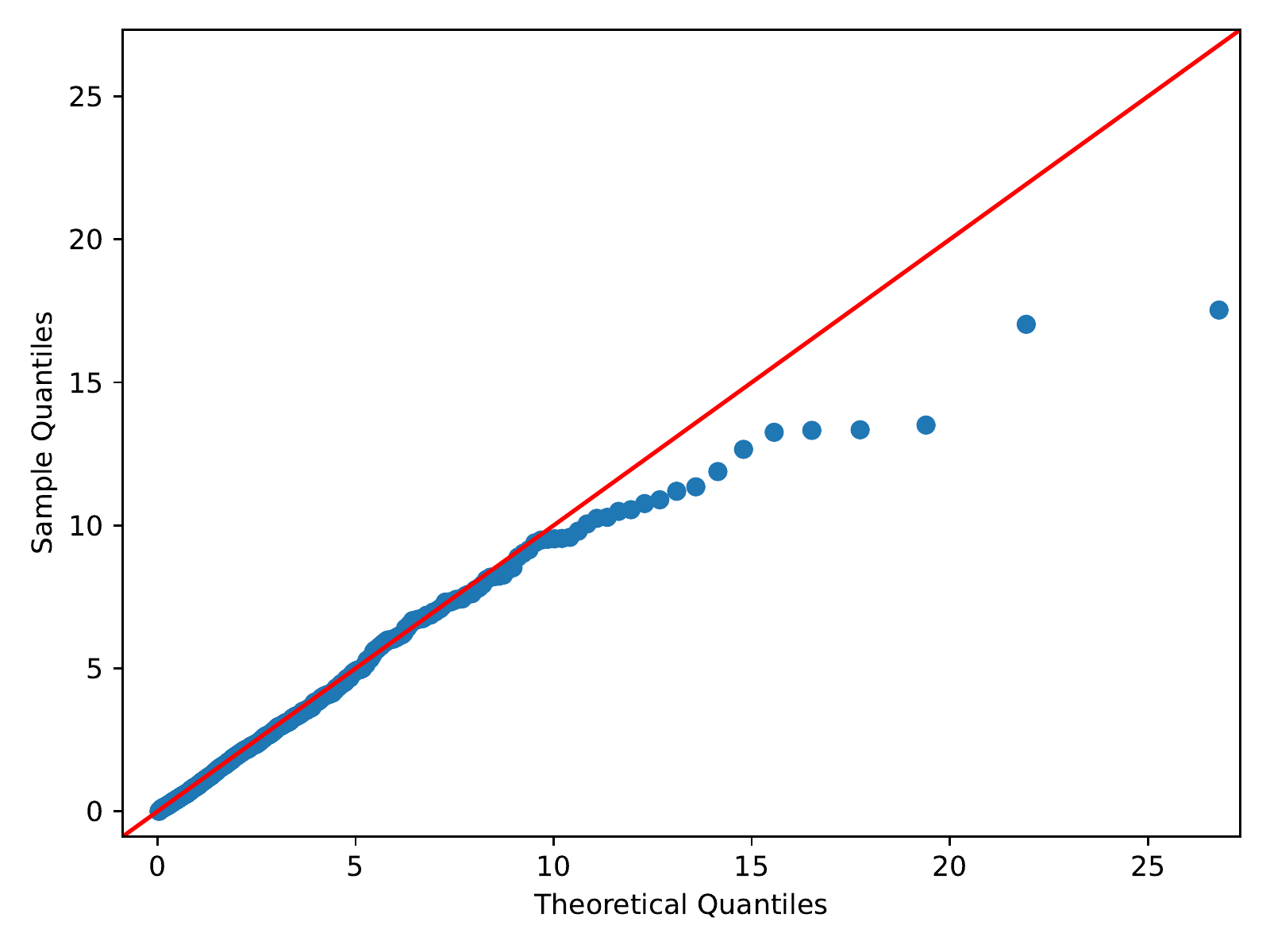}
    \caption{Comparison of QQ plots for Poisson process, Weibull, Gamma, and log-normal renewal processes. The top row are results from our excursion-based estimator. The bottom row are results from a intensity-based estimator. All are fit with the same architecture, optimization parameters and 200 samples from the specified renewal process.}
    \label{fig:qq_baseline}
\end{figure*}
We consider four canonical renewal processes with parameters described in Table~\ref{tab:renewal}.
The architecture is a basic MLP with width 16 and 6 layers and the \texttt{Softplus} activation function.
They were trained for 2000 epochs using the \texttt{AdamW} optimizer with learning rate of $1 \times 10^{-3}$.
The learning rate for the $\delta$ parameter was $1 \times 10^{-2}$.
For reference on how the proposed estimator is performing compared to a standard deep learning-based point process estimator, we provide~\autoref{fig:qq_baseline} which compares the performance of an intensity-based model~\citet{shchur2019intensity} to our excursion-based model.

\begin{table}[]
    \centering
    \begin{tabular}{ll}
    Distribution & Parameters \\ \toprule
        Exponential & $\lambda =1$ \\
        Gamma & $\alpha=9, \, \beta =1$ \\
        Log-Normal & $\mu=0, \,\sigma=1$ \\
        Weibull & $\lambda=1, \, k=1.5$
    \end{tabular}
    \caption{Renewal distributions with the corresponding parameters for the experiments in Figures~\ref{fig:qq}and~\ref{fig:qq_baseline}.}
    \label{tab:renewal}
\end{table}

\subsection{Real Data}  \label{subsec:real_data}

For the real data experiment, we consider a modification of the neural network that includes a positional encoding layer. 
This allows the network to learn higher frequency functions~\citep{tancik2020fourier}.
We use a \texttt{LeakyReLU} activation function. 

\begin{figure}
    \centering
\includegraphics[width=0.48\textwidth]{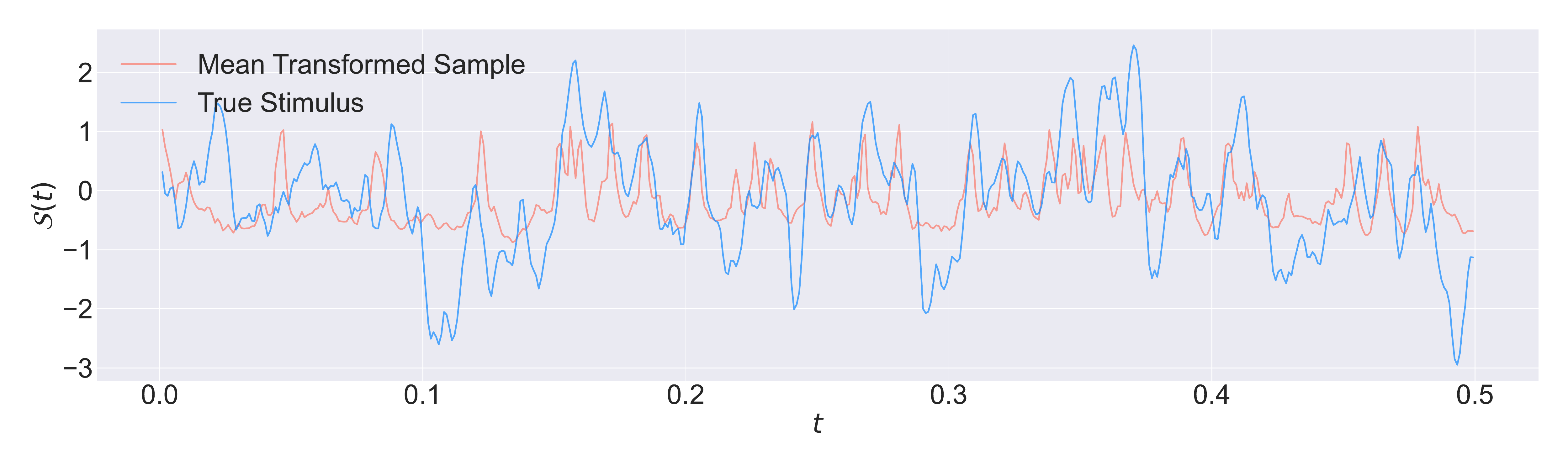}
\includegraphics[width=0.48\textwidth]{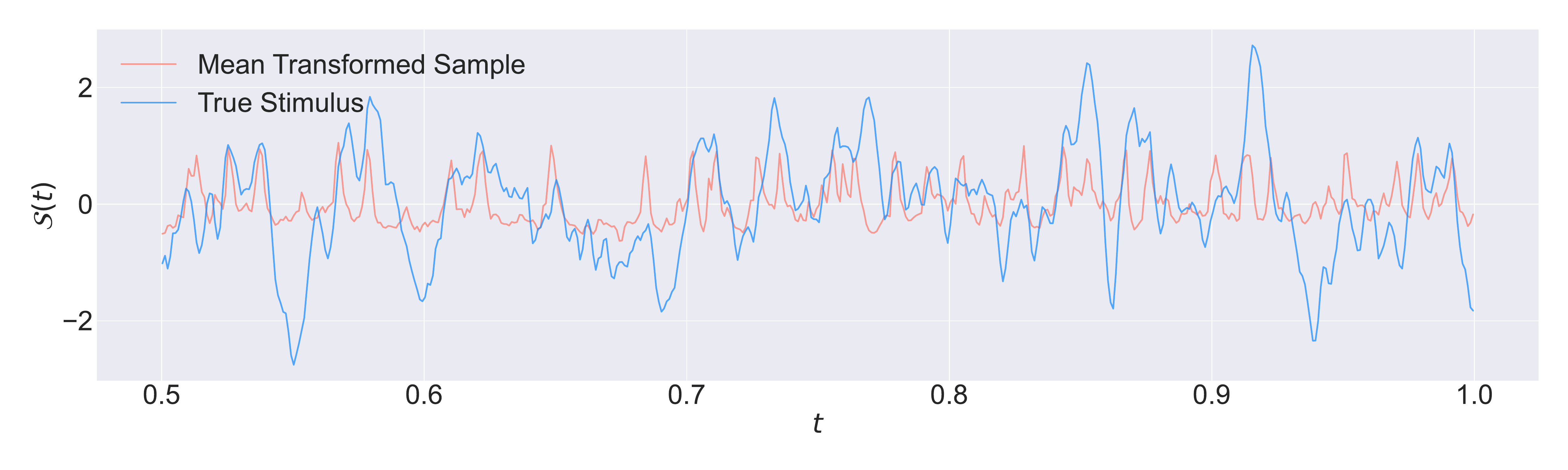} 
\includegraphics[width=0.48\textwidth]{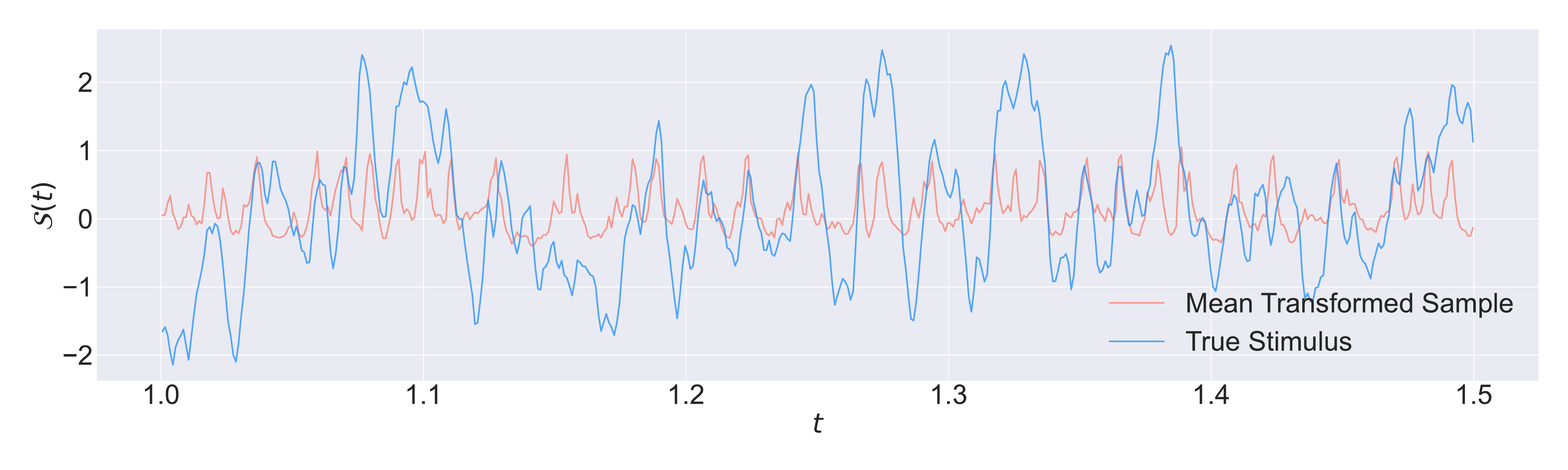}
\includegraphics[width=0.48\textwidth]{figs/learned_stim_4.pdf}
    \caption{Scaled learned stimulus for all time points.}
    \label{fig:learned_stim_all}
\end{figure}

\paragraph{Dataset Description}
The neuroscience dataset is the same as \citep{tripathy2013intermediate}, which is composed of \textit{in vitro} whole-cell patch clamp recordings of mitral cells from mouse olfactory bulb slices.
Spikes were recorded with 100 trials of a 2 seconds duration by simulating the neuron with repeated frozen noise current, which is generated by convolving a white-noise with an alpha function with $\tau = 3$ ms.

\paragraph{Transforming the Learned Stimulus}

Since the proposed method models the continuous as excursions, we invoke a transformation such that the zeros are peaks. 
In particular, we compute $\tilde{Z_t} = a(\mathbb{E}[\log Z_t]) + b$ where $a, b$ are found using least squares. 
Note that due to the $\log$, this will result in $\tilde{Z_t}$ being either positive or negative.
In that sense, we cannot recover the sign of the excursion unless we're provided additional information. 
However, the main components we want to match are the peaks of the true stimulus and the learned stimulus. 
Figure~\ref{fig:learned_stim_all} shows that this is effectively accomplished. 

\section{Further Modeling Considerations}

\subsection{Incorporating History} \label{subsec:history_dependent}
\begin{figure}[h]
\centering
\includegraphics[width=0.4\textwidth]{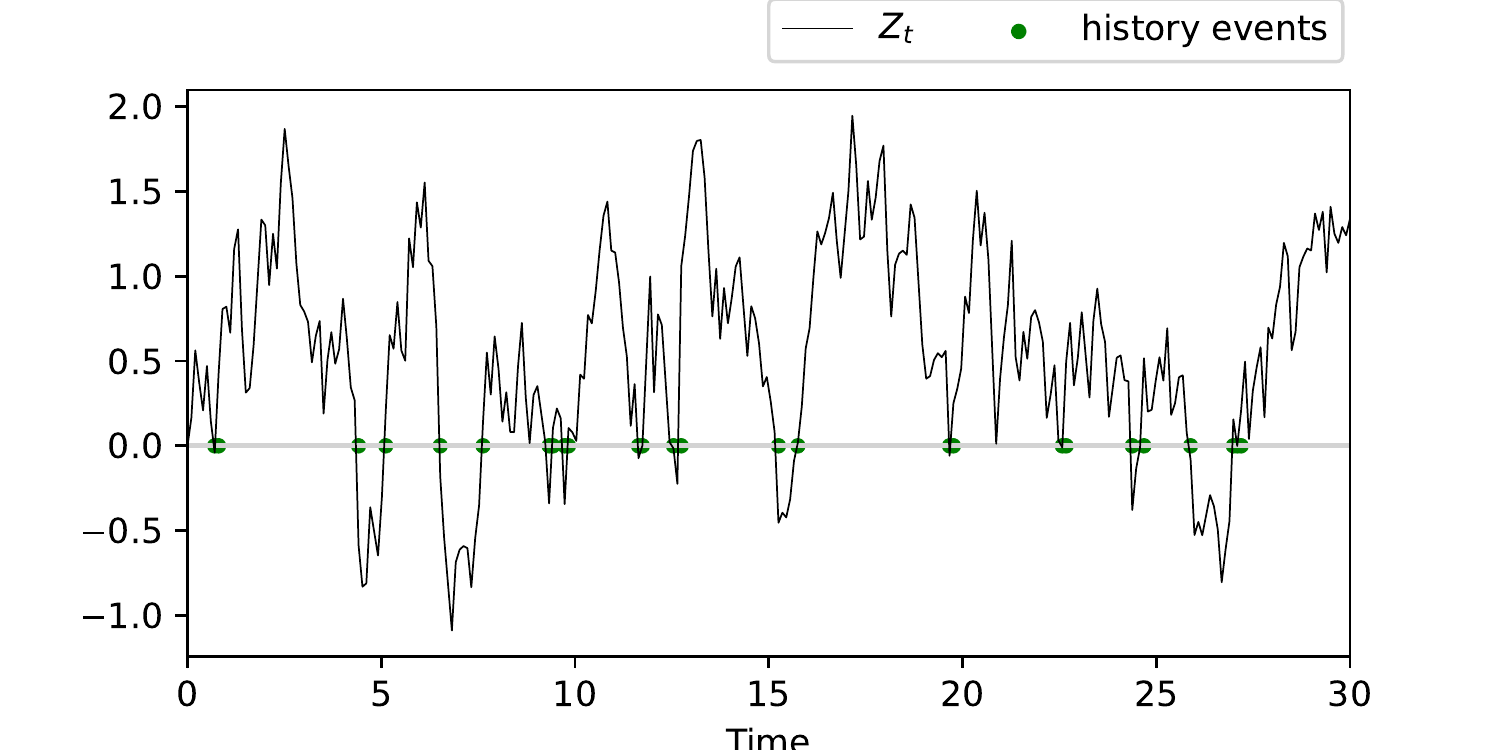}
\caption{Example of a sample path affected by its history.}
\label{fig:history_path}
\end{figure}

In practice, many point processes depend on the distribution of the history for a certain window.
We can include this in the model by making the drift dependent on the history by, for example, considering the structure of the drift to be a recurrent neural network.
Then the likelihood becomes a function of the history with the change of measure being the same as in~\eqref{eq:com}.
Specifically, we can change the drift function from $\mu(Z_t, t)$ in \eqref{eq:sde} to
$\mu(Z_t, t, \mathcal{H}_t)$, where $\mathcal{H}_t$ represents the history of the process up to time $t$. $\mathcal{H}_t$ can be the history of the whole path of $Z_t$, or it can be the history of discrete events 
$\{s \;|\; d L^{f}(s) = 1, s < t \}$, or subset of the history events.
The challenge of the model extension raises from more complicated modeling of the drift function. We leave such extension as the future work.

To illustrate the point, here we consider a special example of history-dependent SDE, which mimics the intensity function of the Hawkes process,
\begin{align}
d Z_t =& \mu(t, \mathcal{H}_t) d t + d W_t \\
\mu(t, \mathcal{H}_t) =& \mu_0 t + h(t - t_{h,1} ) + h(t - t_{h,2})
\end{align}
where $\mu_0 t$ is the baseline drift with a constant slope, and the drift function depends on the last event $t_{h,1}$ and the last second event $t_{h,2}$, so the diffusion process is history-dependent. $h(\cdot)$ is the kernel function describing how past events affect the future drift. In the example, the process only depends on the past two events before $t$. For example, if $h(t) = e^t \cdot \mathbb{I}(t \geq 0)$, the past events will have a positive influence on the drift, and such influence decays as the events stay further back.
As shown in the examples in Appendix \ref{appendix:connection_intensity}, the linear Hawkes process is not equivalent to this problem, even though both of them consider the additive history effects.
An example of a sample path exhibiting this behavior is given in Figure~\ref{fig:history_path}.

\subsection{Connection to More General Point Processes } \label{subsec:connection_point_process}
In this section, we illustrate the modeling connection between the diffusion-based framework and the intensity-based framework.
The bridge between the two groups of methods is the conditional density of the next event given the history up to the current time.
For the intensity-based model, the conditional density of the next event is~\citep{daley2003introduction},
\begin{align*}
g\left(t \mid t_1, \ldots , t_{n-1} \right)& = \lambda\left(t \mid t_1, \ldots, t_{n-1}\right) \times 
\exp \left \{ - \int_{t_{n-1}}^t \lambda \left(u \mid t_1, \ldots , t_{n-1} \right) \mathrm{d}u \right \}
\end{align*}
for $t \geq t_{n-1}$ and where $\lambda$ is the conditional intensity function, 
$t_i$ are timestamps,
$g$ is the density of the next event.
Reversely, given the density $g$, the intensity function can also be derived. See Appendix \ref{appendix:connection_intensity}.
For the diffusion-based model, suppose the point process is generated by an underlying SDE, where the events are a sequence of first-hitting times that reset the process to the initial level after every hit. 
Let $g(\tau), \tau \in \mathbb{R}_+$ be the density of the next event, assuming the current event occurs at $t=0$ ($g(\tau)$ can be different for every hitting time).
For the next event, there exists an SDE with a time-varying drift function $\mu(t)$ such that its first-hitting-time density is $g$ (again, $\mu(t)$ can be different for every hitting time).
The equivalent model is the driftless diffusion, such as Brownian motion, but with a time-varying boundary (the boundary can be different for every hitting time).
If the drift function or the boundary depends on the history, the model will include history dependency.
We leave more detailed discussions and numerical examples, such as the renewal process and Hawkes process, in Appendix \ref{appendix:connection_intensity}.
\subsection{Interpreting the Learned Drift}
A byproduct of the proposed method is a parameterization of a diffusion that generates the data. 
In particular, the drift that is recovered can be subsequently analyzed using traditional methods from stochastic calculus.
In general, the method should be used in cases where such interpretation is useful, such as in health care or in finance. 
The drift may provide clinically useful insight into the distribution of action potentials for a diseased versus healthy patient.
The drift provides an understanding at a multiscale level, leading to potential therapies that correspond to modulating the drift function.

In a similar example as described in the introduction, bursty transcription could be investigated again at a finer, molecular level.
The drift can correspond to some production rate or movement of the underlying molecules which can again lead to potential pathways for developing medications.

Finally, in a financial setting, arrivals of aks orders in a market may be related to excursions of the drawdown process of the perceived fair price. 
As the fair price reaches the running maximum, the drawdown process reaches zero, signaling an appropriate time to sell the financial instrument.
Developing and analyzing the drift of the latent fair price could lead to better risk management for market makers or more effective trading strategies. 

\section{Example Applications}
\label{sec:examples}
To motivate the proposed model, we include a few examples where the method may be appropriate and an intuitive interpretation is present.
We study some of these experiments in greater detail by considering the representational capabilities of the method in describing the data.

\paragraph{Fair Pricing from Bids and Asks of Illiquid Assets.}
Suppose we observe sets of point processes for a given asset on an exchange. 
Denote the set of samples associated to bids as $\mathbb{B} = \{t_i^{(b)}, b_i\}_{i=1}^N$ and the process for asks as $\mathbb{A} = \{t_i^{(a)}, a_i\}_{i=1}^M$. 
We assume the following properties of observations based on a latent fair price, $Z_t$:
\begin{enumerate}
    \item Bids are generated when $Z_t < \mathbb{E}\left[Z_t \, \bigg | \, Z_{t_i^{(b)}} = b_i\right] - \delta$ for $t_i^{(b)} < t$. 
    That is, the fair price should not exceed the expected fair price following the last bid.
    \item The fair price does not cross the expectation of the diffusion bridge between any two arrivals.
\end{enumerate}
$Z_t$ satisfies an SDE with unknown drift $\mu(Z_t, t)$, we consider the excursions above and below the curves:
\\
\begin{tabularx}{\textwidth}{XX}
{\begin{align*}
    f_\text{bid} &= \mathbb{E}[Z_t \mid Z_{t_i^{(b)}} = b_i] \\
    &= \int_{t_i^{(b)}}^t\mu(z, t) dt + b_i 
\end{align*} 
}&{
\begin{align*}
        f_\text{ask} &= \mathbb{E}[Z_t \mid Z_{t_i^{(a)}} = a_i] \\
    &= \int_{t_i^{(a)}}^t\mu(z, t) dt + a_i
\end{align*}
}
\end{tabularx}

Specifically, since $a_i \geq b_i$, $ f_\text{ask}(t) \geq f_\text{bid}(t)$ for all $t >0$. 
Qualitatively, the model suggests that a \emph{new bid} occurs when the fair price exceeds the expected value conditioned on starting at the \emph{last bid} price and the fair price $Z_t$ being at least $\delta > f_\text{bid}(t)$ within that interval with the opposite characterization for the process involving asks.
This is characterized by a $\delta$-excursion above $f_\text{bid}$ or below $f_\text{ask}$, leading to the interarrival time being the excursion length of the fair price above or below the curve. 
We illustrate this behavior and the generated point process in Figure~\ref{fig:bidask}.

\begin{figure}[h]
\vspace{-10pt}
\centering
\includegraphics[width=0.4\textwidth]{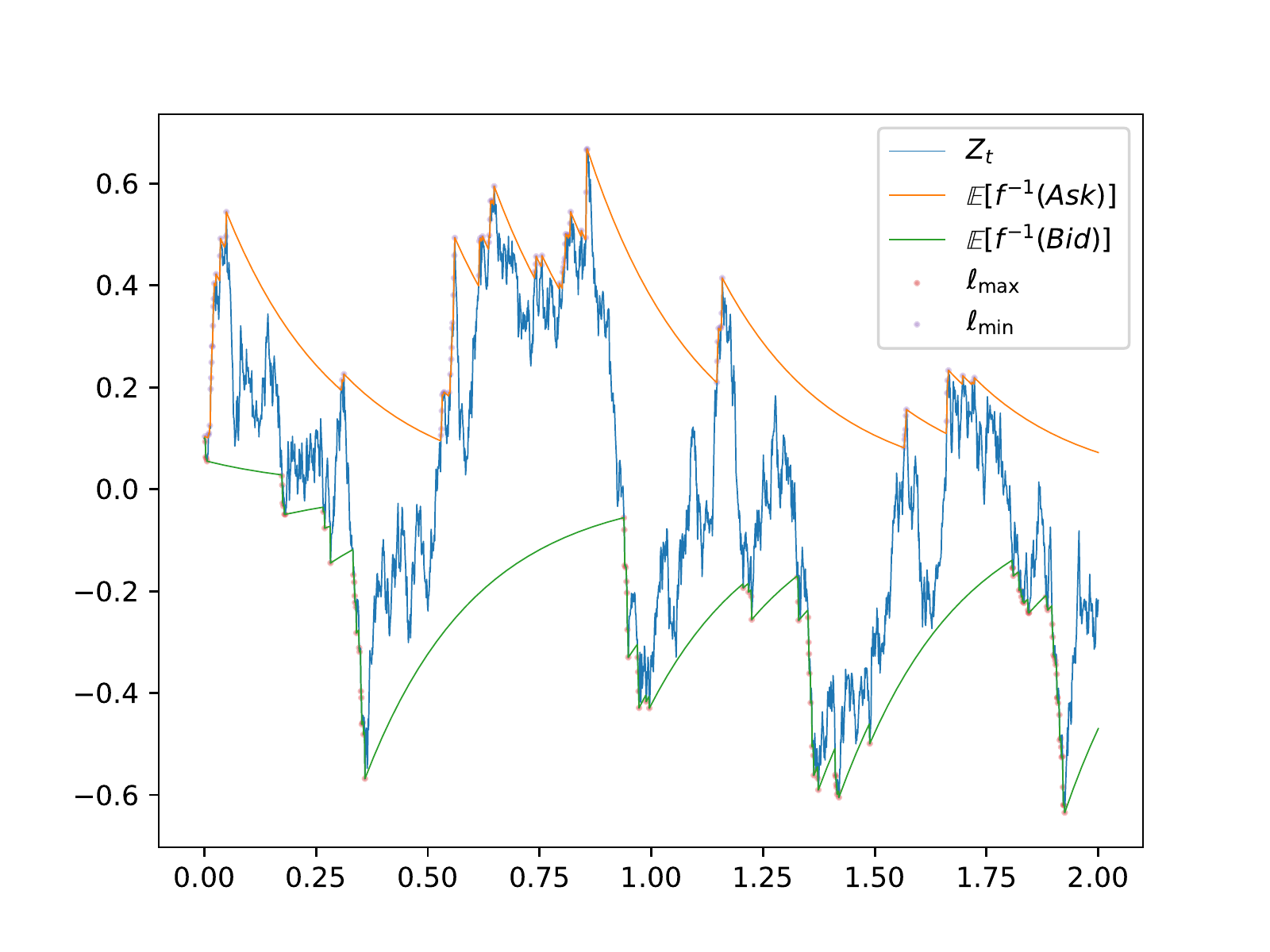}
\caption{Examples of bids and asks generated according to the intersection of the fair price (blue) with the expected price from the last bid (green) and ask (orange).
}
\label{fig:bidask}
\end{figure}

\paragraph{Heart Rate Variability}
In~\citet{barbieri2005point}, the authors consider a renewal process for modeling heart rate variability based on a drifted Brownian motion.
We consider a similar model but based on excursions reaching a minimum height. 
This allows us to make an interpretation of the latent path being a continuous path instead of discontinuous as in the first hitting time case. 
The continuous path should resemble the ensemble of cardiac action potentials that generate the electrical signal whose peaks are the arrival times we observe. 
Analyzing the drift that governs the excursions can then be used for qualitative reasoning about the distributions of action potentials. 

\paragraph{EDA Data}
On another front considering electrophysiology, ~\citet{subramanian2020point} describes a point process structure that describes the spiking properties in electrodermal activity (EDA) data. 
The authors posit the existence of a latent Brownian motion with first crossing times that are distributed according to the observed data. 
The Brownian motion is seen to provide a description of the underlying electrophysiology related to the release of sweat that in turn generates the points in the point process. 
We generalize this argument by considering a series of excursions above a threshold limit. 
This provides a continuous description relating to the release of sweat that generates the arrival of the point process.
The data are in the form of the arrival time of the pulse times derived from EDA data for the control group of patients (\textsc{EDA}), and for the patients under the influence of the drug Propofol (\textsc{EDA-P}). 

\section{Connection Between Diffusion-Based Model and Intensity-Based Model} \label{appendix:connection_intensity}

As introduced in section~\ref{subsec:connection_point_process}, the conditional intensity of the next event relates the two groups of modeling methods of the diffusion-based method and the intensity-based method.
Recall the relationship between the joint arrival time distribution and the intensity function is given by
$$
\lambda(t_n | t_1, \ldots, t_{n-1}) = \frac{p(t_n | t_1, \ldots, t_{n-1})}{1 - \int_{t_{n-1}}^{t_n}p(s | t_1, \ldots, t_{n-1})ds}.
$$
Additionally, we have that $p(t_n | t_{i < n}) = p(t_n, \ldots, t_1)/p(t_{n-1}, \ldots, t_1)$.
Numerically, when considering the proposed estimator, we must compute the ratio of two expectations given our observations. 
The variance of this becomes a bit high, but we can use common random numbers to reduce this. 
That is, we compute the same set of bridges for both expectations.
This leads to only computing the expectation of the martingale portion multiplied by the original measure.
Specifically, 
$$
\lambda(t_n | t_1, \ldots, t_{n-1}) =\lambda(t_n | t_{n-1}) = \frac{p(t_n - t_{n-1} | t_{n-1})}{1 - \int_{t_{n-1}}^{t_n}p(t_n - t_{n-1} | t_{n-1})}.
$$
Of course, if $\mu$ does not depend on $t$ then the conditional term is not needed.
We then approximate the integral in the denominator by a Riemann sum. 

On the other hand, the density of the next event can be derived from the intensity function.
\begin{equation} \label{eq:intensity_to_fpt}
p_n(t | t_1,..., t_{n-1}) = \lambda(t | t_1,..., t_{n-1}) 
\exp\left\{ - \int_{t_{n-1}}^t \lambda(u | t_1,..., t_{n-1}) \mathrm{d}u \right\}.
\end{equation}

Next, we provide a few examples.
Let $p(x, t | x_{t_{(n-1})}, 0)$ be the transition density of the dynamics $X_t$ start from the origin,
$g(u | x_{t_{(n-1})}, 0)$ be the density of the first hitting time of the process. 
Consider the random walk from $(x_{t_{(n-1})}, 0)$ to $(x_t, t)$ through first hitting point $(S(u), u)$.
By marginalizing out the first hitting time $(S(u), u)$ (Chapman–Kolmogorov equation), we have
\begin{equation}
p(x, t | x_{t_{(n-1})}, 0) = \int_0^t p(x, t | S(u), u) \cdot g(u | x_{t_{(n-1})}, 0) \mathrm{d}u.
\label{eq:fortet}
\end{equation}
The above is called the Fortet equation, which can be used to solve the density of the first hitting time \citep{sacerdote2013stochastic}. $g(u | x_{t_{(n-1})}, 0)$ depends on the design of $S(u)$.
In general condition, there exists such a boundary \citep{potiron2021existence}, which is proved using a piecewise linear representation that is also the technique used by \citep{sacerdote2003threshold}. The approximation accuracy can be arbitrarily high (the error is $O(h^2)$ where $h$ is the knot distance) \citep{zucca2009inverse}.

Take the target density $g$ into Eq \eqref{eq:fortet} to solve the boundary function $S$.
If $g$ includes the history of the process, such as the Hawkes process, the solution $S$ becomes a history-dependent and time-varying boundary function.
$p(x, t | x_{t_{(n-1})}, 0)$ and $ p(x, t | S(u), u)$ can be easily obtained using the basic property of Brownian motion.
\citep{sacerdote2003threshold} provides an approximation solution to the problem.

\section{Lamperti Transform} \label{sec:lamperti}
In this section, we show that if $\sigma$ is not a constant, the process can be transformed into an equivalent process with constant $\sigma$. So model in the main text  assumes the constant drift $\sigma = 1$ without loss of generality.
The transformation is described in \citet{ait2002maximum}.
Let $Z_t$ represent the original process satisfying \eqref{eq:sde}.
Let $Y_t := \gamma(X_t,t) = \int \frac{dx}{ \sigma(X_t, t)} $. 
From Ito's lemma,  
\begin{align*}
    dY_t &= \mu_Y(Y_t, t)dt + dW_t \\ 
    \mu_Y(Y_t, t) &= \frac{\partial \gamma}{\partial t} \left (\gamma^{-1}\left(Y_t, t\right), t\right) + \frac{\mu\left(\gamma^{-1}(Y_t, t)\right)}{\sigma\left(\gamma^{-1}\left(Y_t, t\right)\right)} - \frac12 \frac{\partial \sigma}{\partial t} \left ( \gamma^{-1}\left (Y_t, t \right ), t \right ) 
\end{align*}
has unit diffusion. 
To overcome identifiability issues, we assume that this transformation implicitly occurs, and we only wish to recover the unit variance process $Y_t$. 

The multivariate Lamperti transform is described  by \citet{ait2008closed}.
Let $\sigma(y)$ be a symmetric positive matrix for all $y \in \mathbb{R}^d$. 
Additionally let $\sigma$ be differentiable everywhere and for all $y$
$$
\frac{\partial \sigma(y)}{\partial y_k}\sigma(y)^{-1}e_j = \frac{\partial \sigma(y)}{\partial y_j}\sigma(y)^{-1}e_k
$$
where $e_j$ is the $j^{\text{th}}$ canonical basis vector of $\mathbb{R}^d$.
Then the SDE given by Lipshitz $\mu, \sigma$ is a reducible SDE with unit diffusion.

\end{document}